\documentclass[english,aps,showpacs,nobibnotes,superscriptaddress,floatfix, onecolumn,notitlepage,prl,longbibliography, reprint]{revtex4-2}
\usepackage[T1]{fontenc}
\usepackage[latin9]{inputenc}
\usepackage{color}
\usepackage{mathrsfs}
\usepackage{bm}
\usepackage{amsmath}
\usepackage{amsthm}
\usepackage{amssymb}
\usepackage{graphicx}

\makeatletter
\usepackage{braket}
\usepackage{bm}
\usepackage{graphicx}
\usepackage[usenames,dvipsnames]{xcolor}
\usepackage[colorlinks=true,citecolor=Blue,linkcolor=RubineRed,urlcolor=Blue]{hyperref}

\date{\today}

\makeatother

\providecommand\lemmaname{Lemma}
\theoremstyle{plain}
\newtheorem{lem}{\protect\lemmaname}
\usepackage{babel}
\begin{document}
\title{Probe Signatures of Quantum Chaos with Quantum Sensing}
\author{Isa\'ias Vallejo-Fabila}
\address{Institute of Fundamental and Transdisciplinary Research, Institute
of Quantum Sensing and Institute for Advanced Studies in Physics,
Zhejiang University, Hangzhou 310027, China}
\author{Jing Yang}
\email{jing.yang.quantum@zju.edu.cn}

\address{Institute of Fundamental and Transdisciplinary Research, Institute
of Quantum Sensing and Institute for Advanced Studies in Physics,
Zhejiang University, Hangzhou 310027, China}
\address{Nordita, KTH Royal Institute of Technology and Stockholm University,
Hannes Alfv\'ens vag 12, 106 91 Stockholm, Sweden}
\begin{abstract}
The spectral form factor (SFF), which is the Fourier transform of
the two-point spectral correlation, is a fundamental diagnostics of
quantum chaos. Nevertheless, its experimental detection in quantum
many-body systems was less unexplored until recently. Here, we establish
a relationship between the SFF with the quantum Fisher information
(QFI) by means of quantum sensing protocols. Furthermore, when the
eigenvalues and eigenstates of the disordered Hamiltonian are statistically
independent, we prove that QFI is typical both in short and long timescales
while numerics suggest that it is also typical at intermediate timescales,
which ensures self-averaging of QFI. The self-averaging property renders
the experimental measurement of QFI with our protocol less noisy in
contrast with the direct measurement of SFF which is known to be non-self-averaging.
Our results reveal an intriguing connection between quantum noise
and the signatures of quantum chaos.
\end{abstract}
\maketitle
\textit{Introduction}.--- Quantum sensing, also known as quantum
metrology, is the subject of exploiting quantum resources such as
coherence and entanglement to improve the precision of unknown parameters~\citep{pezz`e2018quantum,braun2018quantumenhanced,degen2017quantum}.
The subject plays an important role in detecting the new fundamental
physics~\citep{ye2024essayquantum}. Recent advances in quantum sensing
has developed to leverage the many-body interactions between quantum
sensors to further improve the sensitivity~\citep{montenegro2025reviewquantum}.
In many-body quantum sensing, many efforts have been spent to explore
the fundamental and universal precision limit quantified by QFI in
various many-body sensing protocols~\citep{boixo2008quantum,boixo2007generalized,beau2017nonlinear,yang2022variational,chu2023strongquantum,he2023starklocalization,abiuso2025fundamental,puig2025fromdynamical,block2024scalable}.
In particular, it is found that by properly engineering quantum chaos
due to many-body interaction, it can further boost the precision of
sensing unknown signals~\citep{fiderer2018quantum,li2023improving}.

At the fundamental level, QFI can witness multi-partite entanglement
in many-body quantum systems~\citep{pezze2009entanglement,hyllus2012fisherinformation,toth2012multipartite}
and has been shown to be related to the dynamical susceptibility through
the linear response theory~\citep{hauke2016measuring}, which has
inspired lots of efforts~\citep{pezz`e2017multipartite,zhang2018characterization,brenes2020multipartite}.
Here, we reveal another fundamental relation between QFI in quantum
sensing protocols and SFF, the diagnostics of quantum chaos. Given
a quantum many-body system, detecting its chaotic nature has recently
spurred several initial efforts~\citep{vallejo-fabila2025singlesite,das2025proposal,dong2025measuring,daug2023manybody}.
The relation between QFI and SFF we discover here, on one hand, sets
the fundamental and universal precision limit in quantum sensing with
quantum chaotic dynamics, on the other hand, it allows us probe quantum
chaos signatures with quantum sensing methods. We also show both analytically
and numerically that unlike SFF~\citep{prange1997thespectral} or
the survival probability~\citep{schiulaz2019thouless}, the QFI is
self-averaging, which render its experimental detection less noisy.
In contrast with the case of quantum sensing where the signals are
unknown, we propose to couple an unknown many-body Hamiltonian with
\textit{known} signals and experimentally measure QFI through proper
measurement techniques. From the foundational perspective, since the
precision of quantum sensing is dominated by quantum noise, our results
reveal the fundamental connection between quantum noise and quantum
chaos signatures. 

\textit{Setup}.--- From an experimental perspective, we would like
to probe the chaotic nature of a many-body quantum system, clean or
disordered, at a given time $T$. Specifically, we would like to know
whether the spectrum of the system Hamiltonian is integrable, satisfying
Poisson statistics or chaotic satisfying Wigner-Dyson level statistics,
predicted by the Gaussian ensembles of random matrix theory (RMT)~\citep{haake2001quantum,mehta2004randommatrices}.
However, in experimental platforms such as cold atoms and ion traps,
one cannot directly access to the energy spectrum~\citep{Schreiber2015,Hess2017}.
Here, inspired by the spirit of quantum metrology, we consider to
apply a sensing Hamiltonian $\Pi_{\bm{\theta}}=\sum_{\alpha}\theta_{\alpha}P_{\alpha}$
for a time $T$, the coupling parameter $\theta_{\alpha}$ is known
and can be carefully tuned, and the coupling parameters $P_{\alpha}$
can be engineered. Upon time evolution, we expect quantum chaos signatures
to be encoded in the parameter-dependent states, which then can be
further probed through quantum measurements. \textcolor{black}{Two
protocols naturally arise: the continuous protocol, where system and
sensing Hamiltonians are applied at the same time, and }the quench
protocol, where the system Hamiltonian is applied for a time $T$
followed by the sensing Hamiltonian applied for the same time $T$.
From the perspective of quantum sensing, the system Hamiltonian $H$
can be viewed as a many-body control to create necessary entanglement
for sensing. 
\begin{figure}
\centering
\includegraphics[scale=0.125]{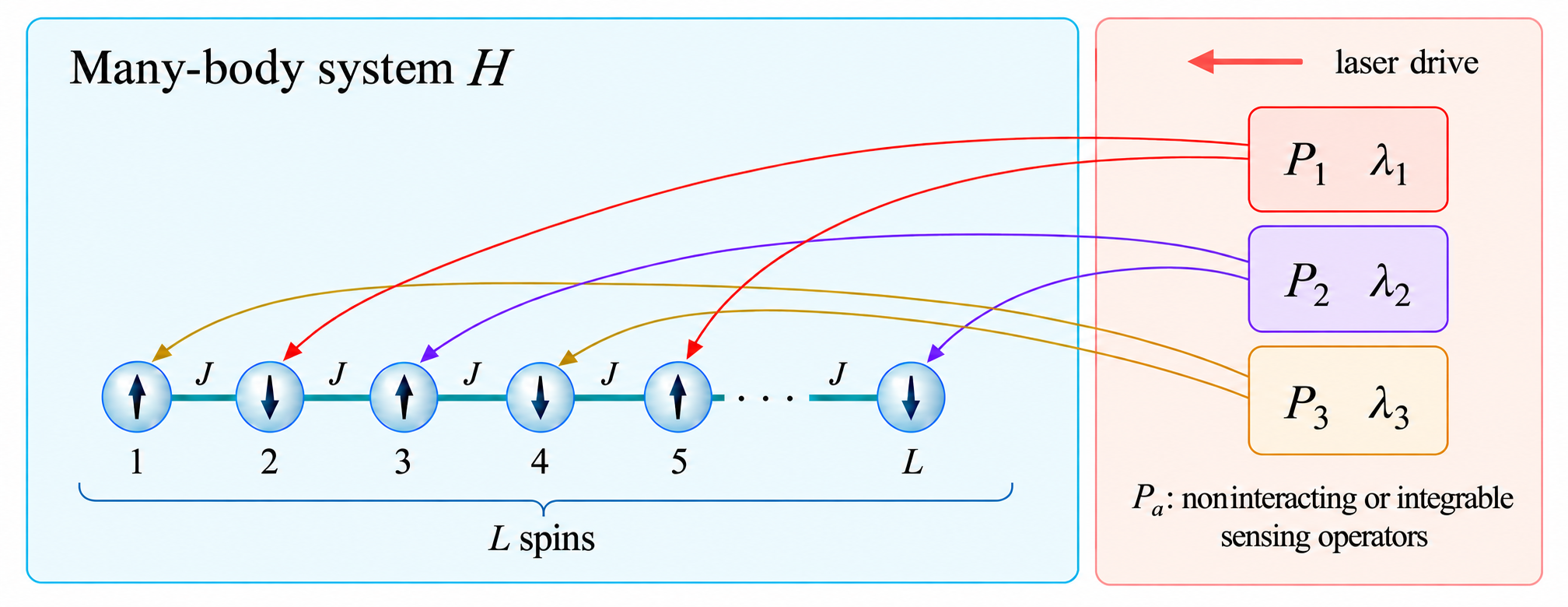}

\caption{The setup of probing a many-body quantum system with a sensing Hamiltonian
$\Pi_{\bm{\theta}}=\sum_{\alpha}\theta_{\alpha}P_{\alpha}$, where
signals $\bm{\theta}$ are known. The initial state is chosen from
a Haar random ensemble. The QFI tomography is performed at the single-realization
level. The ensemble average of QFI gives rise to the time integral
of SFF, whose fluctuations are suppressed in the large-system limit. }
\end{figure}
The continuous protocol allows entanglement creation and sensing are
performed concurrently while the quench protocol exploits the control
to create entanglement first and perform the sensing tasks. For $t\in[0,T]$,
the\textcolor{black}{{} total Hamiltonian for the continuous protocol
and }quench protocol\textcolor{black}{{} are}
\begin{eqnarray}
H_{\bm{\theta}} & = & H+\Pi_{\bm{\theta}},
\end{eqnarray}
and
\begin{equation}
H_{\bm{\theta}}=[1-\Theta(t-T)]H+\Theta(t-T)\Pi_{\bm{\theta}},
\end{equation}
respectively, where $\Theta(t)$ is the Heaviside function. In addition,
the geometric tensor (GT) for both protocols, with given initial state
and many-body Hamiltonian $H$ is
\begin{equation}
Q^{\mathrm{ctn}}_{\alpha\beta}\equiv4\mathrm{Cov}^{>}[G^{\mathrm{ctn}}_{\alpha},\,G^{\mathrm{ctn}}_{\beta}]{}_{\ket{\psi_{0}}},Q^{\mathrm{qch}}_{\alpha\beta}\equiv4\mathrm{Cov}^{>}[G^{\mathrm{qch}}_{\alpha},\,G^{\mathrm{qch}}_{\beta}]{}_{\ket{\psi_{0}}},\label{eq:Q-def}
\end{equation}
where the metrological generators are defined as $G^{\mathrm{ctn}}_{\alpha}\!\equiv\!\int^{T}_{0}e^{\mathrm{i}H_{\bm{\theta}}s}P_{\alpha}\mathrm{e}^{-\mathrm{i}H_{\bm{\theta}}s}ds,$
$G^{\mathrm{qch}}_{\alpha}\!\equiv\!e^{\mathrm{i}HT}\!\left(\int^{T}_{0}ds\mathrm{e}^{\mathrm{i}H_{\bm{\theta}}s}P_{\alpha}\mathrm{e}^{-\mathrm{i}H_{\bm{\theta}}s}\right)\!e^{-\mathrm{i}HT}$
~\citep{boixo2007generalized}, and the non-symmetric covariance
between two operators over a state is $\mathrm{Cov}^{>}(A,B)\!\equiv\!\braket{\psi_{0}|AB|\psi_{0}}\!-\!\braket{\psi_{0}|A|\psi_{0}}\braket{\psi_{0}|B|\psi_{0}}.$
The real part of the GT is known as the QFI while the imaginary part
is known as the Berry curvature~\citep{kolodrubetz2017geometry,carollo2020geometry,provost1980riemannian}.
In the small-signal limit $\bm{\theta}\to0$ the generators for both
protocols simplify,
\begin{equation}
G^{\mathrm{ctn}}_{\alpha}\approx\int^{T}_{0}ds\mathrm{e}^{\mathrm{i}Hs}P_{\alpha}\mathrm{e}^{-\mathrm{i}Hs},\,G^{\mathrm{qch}}_{\alpha}\approx T\mathrm{e}^{\mathrm{i}HT}P_{\alpha}\mathrm{e}^{-\mathrm{i}HT}.
\end{equation}

\textit{Random matrix prediction from isospectral ensemble.} ---
Our goal is based on the average of the GT for Haar-random initial
states and the Hamiltonian ensemble (the second is trivial for clean
Hamiltonians) and connect them with the SFF, a key quantity in the
witness of quantum chaos. In this section, we consider the Hamiltonian
ensemble as the isospectral ensemble where all the eigenstates are
Haar-random and the eigenvalues are identical and follows the same
probability distribution. We denote the isospectral ensemble as $\mathrm{ISO}\!=\!\{H=U\Lambda U^{\dagger}|U\sim\mathrm{Haar},\,\Lambda=\mathrm{diag}(\lambda_{1},\lambda_{2},\cdots,\lambda_{N})\sim d\mu(\Lambda)\},$where
$d\mu(\Lambda)$ is the probability measure of the eigenvalues. The
Gaussian Unitary Ensemble (GUE) is a special case of the isospectral
ensemble~\citep{mehta2004randommatrices} with $d\mu(\Lambda)=\mathrm{e}^{-\sum^{N}_{i=1}\lambda^{2}_{i}}\prod_{1\leq j<k\leq N}|\lambda_{j}-\lambda_{k}|^{2}\prod^{N}_{\ell=1}d\lambda_{\ell}$.
Before we give our main results we introduce the notation. It is clear
from Eq.~(\ref{eq:Q-def}) that $Q_{\alpha\beta}$ is a function
of $H$ and $\ket{\psi_{0}}$. The average of the GT over the initial
Haar-random states is $\overline{Q}_{\alpha\beta}\equiv\mathbb{E}_{\ket{\psi_{0}}\sim\mathrm{Haar}}[Q_{\alpha\beta}(\ket{\psi_{0}})]$
and its average over the Hamiltonian ensemble is $\langle\!\langle Q_{\alpha\beta}\rangle\!\rangle\equiv\mathbb{E}_{H\sim\mathrm{GUE}}[Q_{\alpha\beta}(H)]$.
Here, the double angle bracket indicates that the average over the
isospectral Hamiltonian ensemble consisting of independent averages
over the eigenstates and eigenvalues. Clearly, $\langle\!\langle\overline{Q}_{\alpha\beta}\rangle\!\rangle$
denotes the average over both, initial Haar-random states and isospectral
Hamiltonian ensemble. Furthermore, averages over the eigenvalue distribution
and the Haar-random distribution in the Hamiltonian ensemble are denoted
as $\langle\bullet\rangle_{\Lambda}$ and $\langle\bullet\rangle_{U}$,
respectively. The average of the isospectral ensemble for a given
initial state can be obtained via the Schur-Weyl duality~\citep{harrow2005applications,goodman2009symmetry,roberts2017chaosand}
and yields both the two-point and four-point SFF~\citep{liu2018spectral,oliviero2021randommatrix,cotler2020spectral},
see Sec.~\ref{sec:Isospectral-twirling} in Supplemental Material~\citep{SM}. 

It turns out that first averaging over the initial Haar-random states
and then over the isospectral ensemble simplifies the results (Sec.~\ref{sec:QFI-closed-form}
in Supplemental Material~\citep{SM}). For the quench protocol, the
average of the GT over the ensemble of Haar-random initial states
is
\begin{equation}
\overline{Q}^{\mathrm{qch}}_{\alpha\beta}(T)=4C_{\alpha\beta}T^{2},\label{eq:Q-qch}
\end{equation}
independent of the system Hamiltonian, while for the continuous protocol,
the average of the GT over the ensemble of Haar-random initial states
and the isospectral Hamiltonian becomes
\begin{align}
\langle\!\langle\overline{Q}^{\mathrm{ctn}}_{\alpha\beta}(T)\rangle\!\rangle & =\frac{4C_{\alpha\beta}}{N^{2}-1}[2N^{2}\int^{T}_{0}(T-t)K(t)dt-T^{2}]\nonumber \\
 & =\frac{4C_{\alpha\beta}T^{2}}{N^{2}-1}\left[2N^{2}I(T)-1\right],\label{eq:Q-ctn}
\end{align}
where $N$ is the dimension of the many-body Hilbert space, the constant
$C_{\alpha\beta}\equiv[N\mathrm{Tr}(P_{\alpha}P_{\beta})-\mathrm{Tr}P_{\alpha}\mathrm{Tr}P_{\beta}]/[N(N+1)],$
the average SFF is defined as
\begin{equation}
K(t)=\frac{1}{N^{2}}\langle\sum_{kl}\mathrm{e}^{\mathrm{i}(\lambda_{k}-\lambda_{l})t}\rangle_{\Lambda}=\frac{1}{N^{2}}\langle\!\langle|\mathrm{Tr}\mathrm{e}^{\mathrm{i}Ht}|^{2}\rangle\!\rangle,
\end{equation}
with $K(0)=1$, and the time-integrated SFF is defined as
\begin{equation}
I(T)\equiv\int^{1}_{0}(1-u)K(uT)du.
\end{equation}
A few observations can be made: (i) It is clear that the imaginary
part of the GT, i.e, the Berry curvature is washed out upon performing
the average over the ensemble of Haar-random initial states and therefore
the average GT reduces to the average QFI. (ii) The quench protocol
cannot detect the integrable or chaotic nature of a many-body quantum
system, the average of QFI over the ensemble of Haar-random initial
states does not depends on the Hamiltonian ensemble as indicated by
Eq.~(\ref{eq:Q-qch}). (iii) When the isospectral Hamiltonian ensemble
is GUE we obtain a closed-form expression for the time-integrated
SFF $I(T)$, $\lim_{T\to\infty}{\displaystyle I(T)}={\displaystyle 1/(2N)}$
and $\lim_{T\to0}{\displaystyle I(T)}={\displaystyle 1/2}$ (see Supplemental
Material~\citep{SM}). The comparison between Eq.~(\ref{eq:Q-ctn})
and numerical calculation is shown in Fig.~\ref{fig:Numerics}(b).

\textit{Universal behavior for realistic Hamiltonian ensembles}. ---
For the continuous protocol, Eq.~(\ref{eq:Q-ctn}) can be computed
analytically in the context of RMT by solving the integral of the
SFF~\citep{Cotler2017}, but for realistic Hamiltonians or Hamiltonian
ensembles it is not easy. If a Hamiltonian is chaotic, only its two-point
spectral correlations are predicted by RMT according to the Bohigas-Giannoni-Schmit
conjecture~\citep{bohigas1984characterization,guhr1998randommatrix}.
The density of states of a local chaotic Hamiltonian near the middle
of the spectrum is Gaussian~\citep{keating2015spectra,brody1981randommatrix}
while the density of states for Gaussian random matrix ensemble is
the celebrated Wigner semi-circle~\citep{mehta2004randommatrices}.
This also leads to the fact that the eigenstates clustered around
the middle of the spectrum of a chaotic Hamiltonian are Haar random
as predicted by RMT. The description of a generic chaotic eigenstate
is given by the random wave model~\citep{berry1977regular,santos2012chaosand,schoeppl2025universal}
and the eigenstate thermalization hypothesis (ETH)~\citep{srednicki1999theapproach,deutsch1991quantum,srednicki1994chaosand,dalessio2016fromquantum}. 

To get around the issue of non-Haar-randomness of the realistic chaotic
ensembles, let us note a simple observation: in the quench protocol,
the state at time $T$ is always a Haar-random state upon randomizing
the initial state according to the Haar ensemble, regardless of the
Hamiltonian ensemble. 
\begin{equation}
d\mu(\ket{\psi(T)})=d\mu(W(T))d\mu_{\mathrm{Haar}}(V)\ket{\phi_{0}}=d\mu(\ket{\psi_{0}}),
\end{equation}
where $W(T)=\mathrm{e}^{-\mathrm{i}HT}$. Inspired by this observation,
we propose a \textit{randomized coupling protocol}: we perform the
twirling operation on the coupling operators, $\{P_{\alpha}\}$, which
are randomized according to $V\{P_{\alpha}\}V^{\dagger}$, with $V$
being Haar-random unitary on the whole Hilbert space. We denote the
QFI averaged over Haar-randomized coupling operators by $\widetilde{Q}_{\alpha\beta}\equiv\mathbb{E}_{V\sim\mathrm{Haar}}[Q_{\alpha\beta}(V\{P_{\alpha}\}V^{\dagger})]$.
It can be shown that the operation of twirling the coupling effectively
turn all the eigenstates of realistic Hamiltonian ensemble to be Haar
random such that (details are in the Supplemental Material~\citep{SM}):
\begin{equation}
\langle\!\langle\widetilde{\overline{Q^{\text{ctn}}_{\alpha\beta}}}(T)\rangle\!\rangle=\langle\!\langle\overline{Q^{\text{ctn}}_{\alpha\beta}}(T)\rangle\!\rangle.\label{eq:Q-ctn-tilde}
\end{equation}
The twirling technique has been widely used in quantum information
theory and, more recently, in quantum metrology~\citep{imai2026referenceframeindependent,xu2025informationally}.
Here we exploit the twirling operation to establish the link between
QFI, the core notion in quantum sensing, and the SFF, a fundamental
quantity in quantum chaos. 

It is then natural to ask whether a similar relation between the averaged
QFI and the SFF exists for deterministic couplings. Here, we assume
the permutation invariance of the distribution of the eigenvalues
and statistical independence of eigenstates and eigenvalues~\citep{schiulaz2019thouless}.
Under these assumptions, we can separate the non-universal part that
depends on the distribution of the eigenstates in the averaged QFI
and to obtain\textcolor{red}{{} }
\begin{equation}
\langle\!\langle\overline{Q}^{\mathrm{ctn}}_{\alpha\beta}(T)\rangle\!\rangle\!=\!\frac{4}{N^{2}-1}[2N^{2}A_{\alpha\beta}\int^{T}_{0}(T-t)K(t)dt\!-\!B_{\alpha\beta}T^{2}],\label{eq:Q-ctn-phy-H}
\end{equation}
where $A_{\alpha\beta}$ and $B_{\alpha\beta}$ are non-universal
constants that depend on the microscopic eigenstate distributions
of the realistic Hamiltonian ensemble and both reduce to $C_{\alpha\beta}$
when the distribution of all eigenstates becomes Haar-random. We conclude
by noting that the connection between averaged QFI and SFF straightforwardly
generalize to the Haar average over orthogonal and symplectic groups\textcolor{red}{~}\citep{collins2006integration,collins2009onsome,collins2022theweingarten,brauer1937onalgebras}
(see Supplemental Material~\citep{SM}).

\textit{Typicality and self-averaging of QFI.}--- For an initial
state drawn from the Haar-random ensemble, due to the concentration
of Haar measure, initial-state typicality is satisfied as the dimension
of the Hilbert space $N$ goes to infinity. Same observation holds
for the typicality of the Haar-randomized coupling operators in the
continuous randomized protocol discussed previously. In this section
we focus on typicality and self-averaging with respect to the Hamiltonian
ensemble. We consider the SFF and its time-integral at the single-realization
level,
\begin{equation}
K_{H}(T)\equiv\frac{1}{N^{2}}|\mathrm{Tr}(\mathrm{e}^{-\mathrm{i}HT})|^{2},\,I_{H}(T)\equiv\int^{1}_{0}(1-u)K_{H}(uT)du,\label{eq:K-I-def}
\end{equation}
where by definition $K(T)=\langle\!\langle K_{H}(T)\rangle\!\rangle$
and $I(T)=\langle\!\langle I_{H}(T)\rangle\!\rangle$. Clearly, the
fluctuations of $Q^{\mathrm{ctn}}_{\alpha\beta}(T)/T^{2}$ over the
Hamiltonian ensemble are governed by the fluctuations of $I_{H}(T)$.
We define the relative variance (RV) for the SFF and the time-integrated
SFF over the Hamiltonian ensemble as $\mathrm{RV}_{N}[K_{H}(T)]=\langle\!\langle I^{2}_{H}(T)]\rangle\!\rangle/I^{2}(T)-1$,
and $\mathrm{RV}_{N}[I_{H}(T)]=\langle\!\langle I^{2}_{H}(T)]\rangle\!\rangle/I^{2}(T)-1$,
respectively. We emphasize that the subtle point that typicality and
self-averaging of a time-dependent quantity, like SFF or its time-integral,
are related but different notions. Self-averaging would require typicality
at all timescales, including infinite or asymptotically large. For
example, it is well known that the SFF for closed systems is not self-averaging~\citep{prange1997thespectral,cipolloni2023onthe}
because after the Heisenberg time $T\gtrsim N$, $\mathrm{RV}_{N}[K_{H}(T)]\sim1$.
Nevertheless, one can show that the SFF is typical at early times.
We also prove that when eigenstates and eigenvalues of the disordered
Hamiltonian are strictly statistically independent, after the Heisenberg
time $T\gg N$, $\mathrm{RV}_{N}[I_{H}(T)]\to0$ as $N\to\infty$
and $I_{H}(T)$ is typical at early times (see Supplemental Material~\citep{SM}).
Note that for the continuous-deterministic protocol, the property
of statistical independence is just an assumption, although it is
able to predict the average QFI accurately as one can see Fig.~\ref{fig:Numerics}(c).
However, whether this assumption can be employed to predict RV is
unclear because numerical confirmation in the thermodynamic limit
is beyond current numerical resources and capabilities (see Sec.~\ref{sec:Numerical-Calculation}
in Supplemental Material~\citep{SM}). 

The lack of self-averaging of SFF implies large fluctuations on its
time evolution and to obtain reliable measurements in experiments
is a challenge, even if all the technological barriers can be overcame.
Here, self-averaging and typicality of QFI allows us to measure the
SFF indirectly to a very high precision. Operationally, one can choose
a set of coupling operators $\{P_{\alpha}\}$ and a set of parameters
$\{\theta_{\alpha}\}$, which can be \textit{known a priori}. Then,
one can perform the ``QFI tomography'' at different timescales to
extract the SFF from the second time derivatives of the averaged QFI
through the following relation
\begin{equation}
d^{2}\langle\!\langle\widetilde{\overline{Q}}^{\mathrm{ctn}}_{\alpha\beta}\rangle\!\rangle/dT^{2}=\frac{8C_{\alpha\beta}}{N^{2}-1}[N^{2}K(T)-1],\,
\end{equation}
\begin{equation}
d^{2}\langle\!\langle\overline{Q}^{\mathrm{ctn}}_{\alpha\beta}\rangle\!\rangle/dT^{2}=\frac{8}{N^{2}-1}[N^{2}A_{\alpha\beta}K(T)-B_{\alpha\beta}].
\end{equation}
Such tomography of the QFI can be implemented in several ways. For
example, one can prescribe a set of \textit{known} weak signals $\bm{\theta}$.
The quantum Fisher information matrix at time $T$ can be measured
in the eigenbasis of the symmetric logarithmic derivative (SLD) $L_{i}(T)=2(\ket{\partial_{i}\psi_{\bm{\theta}}(T)}\bra{\psi_{\bm{\theta}}(T)}+\mathrm{h.c.})$
and then perform the classical parameter estimation with the maximum
likelihood estimator. Then if the number of repetitions of the experiments
$n$ is very large, $[\mathrm{Cov}(\hat{\bm{\theta}}^{\mathrm{est}})]^{-1}/n$
is the QFI with the $\mathrm{Cov}(\hat{\bm{\theta}}^{\mathrm{est}})$
is the covariance matrix of estimating the parameters $\bm{\theta}$.
In many-body systems, the eigenbasis of the SLD can be nonlocal, which
is challenging but still feasible with the current technology. One
can use local measurements with classical communications to reach
the quantum Cram\'er-Rao bound~\citep{zhou2020saturating,liu2025optimal}.
Recent techniques have been also developed to estimate the QFI directly~\citep{zhang2026gapclosing,zhang2025krylovshadow}.

\textit{Examples.}--- We consider three models, random matrices from
GUE and two spin$-1/2$ systems with open boundary conditions, the
disordered full-random-field XXX (FRF-XXX) $H_{\mathrm{FRF-XXX}}=J\sum^{L-1}_{i=1}\bm{S}_{i}\cdot\bm{S}_{i+1}+\sum_{i=1}\bm{h}_{i}\cdot\bm{S}_{i}$,
where $\bm{S}_{i}=(\sigma^{x}_{i},\,\sigma^{y}_{i},\,\sigma^{z}_{i})/2$
and $\bm{h}_{i}=(h^{x}_{i},\,h^{y}_{i},\,h^{z}_{i})$ with $h^{\alpha}_{i}\in[-h,h]$,
and the clean mixed-field Ising (MFI) $H_{\mathrm{MFI}}=J\sum^{L-1}_{i=1}S^{x}_{i}S^{x}_{i+1}+g^{x}\sum^{L-1}_{i=1}S^{x}_{i}+g^{z}\sum^{L-1}_{i=1}S^{z}_{i}+d_{L/2}S^{z}_{L/2}$\footnote{The term $d_{L/2}S^{z}_{L/2}$ is added to break the total $S^{z}=\sum^{L}_{i=1}S^{z}_{i}$
conservation for the convenience of the numerical calculation of the
level statistics}. Numerical results are shown in Fig.~\ref{fig:Numerics}, which
validates our findings between the relation of the average QFI and
the SFF in the continuous protocol. Furthermore, as shown in the insets
of Fig.~\ref{fig:Numerics}(b) and~\ref{fig:Numerics}(c), numerical
results suggest that the fluctuations of the time-integrated SFF are
not only negligible at short and long times, but also at intermediate
times, suggesting that the QFI is self-averaging for RMT and for realistic
chaotic Hamiltonians at all timescales. Note that although the correlation
hole, which is the characteristics of quantum chaos, disappears in
QFI due to time integration, one can always recover it upon reconstructing
the SFF. 
\begin{figure}
\centering
\includegraphics[scale=0.31]{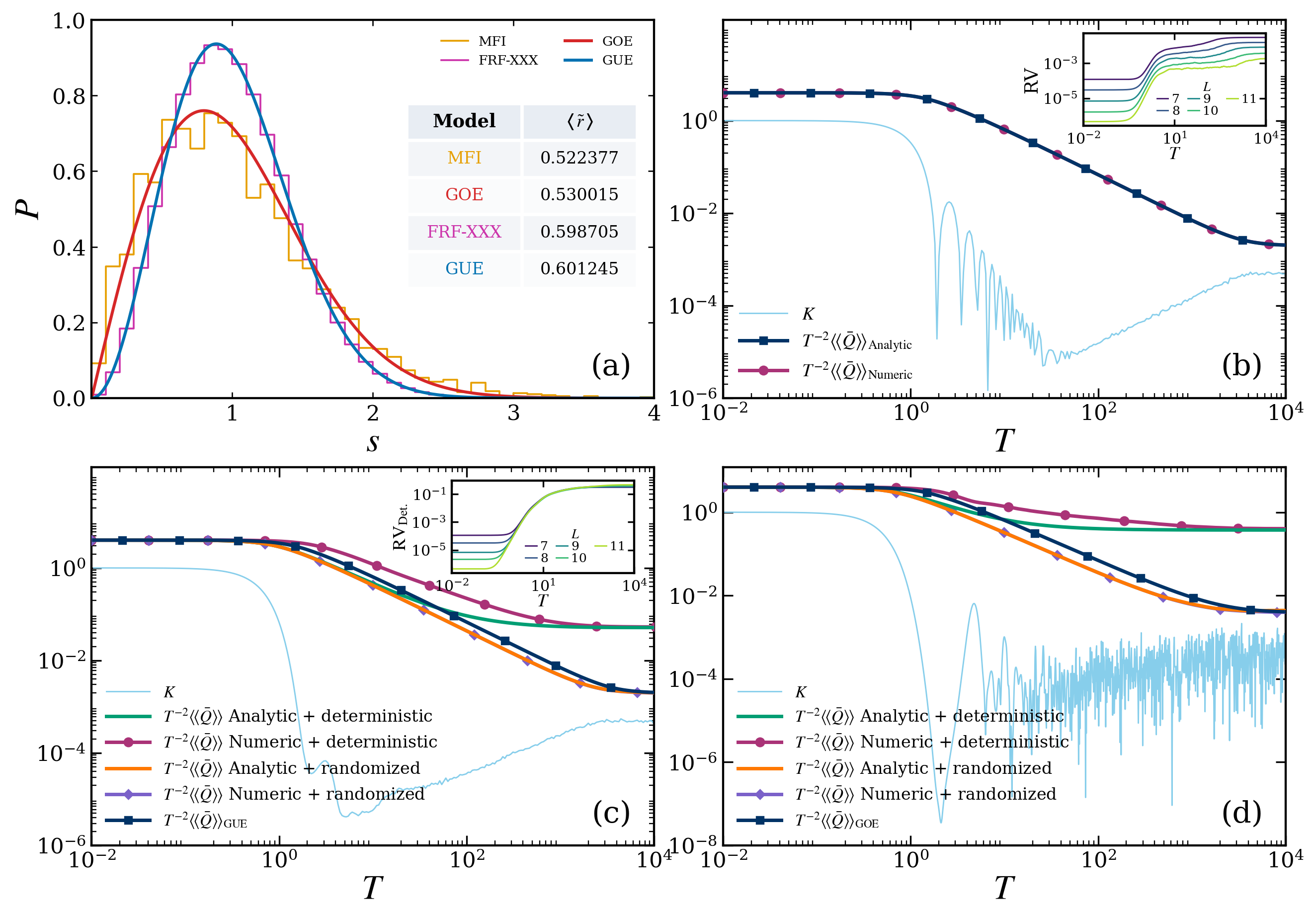}\caption{\label{fig:Numerics}\textcolor{blue}{{} }We consider a single-parameter
estimation problem with $\Pi_{\theta}=\theta P$, $P=\sigma_{x}\otimes\mathbb{I}^{\otimes(L-1)}$
and $N=2^{L}$. (a) Numerical calculations of the mean-level spacing
distribution and the mean adjacent gap ratio ensure that the FRF-XXX
and the MFI models are in the chaotic regime, satisfying GUE- and
GOE-like level statistics, respectively. (b) For GUE, the analytical
results, dark blue with squares based on Eq.~(\ref{eq:Q-ctn}), agrees
with numerics, dark pink with circles computed through the variance
of the generator. (c) For the disordered FRF-XXX model, two subprotocols
are considered: the deterministic protocol based on Eq.~(\ref{eq:Q-ctn-tilde})
and the randomized protocol based on Eq.~(\ref{eq:Q-ctn-phy-H}).
(d) For the clean MFI model, analytics for the two subprotocols are
based on Sec.~\ref{sec:Orthogonal-Groups} of ~\citep{SM}. The
insets in (b) and (c) is the relative variance of the QFI for the
GUE and deterministic protocol for the FRF-XXX, respectively. The
relative variance in the randomized protocol displays similar behavior
with the inset in (b). The parameter values are: $\theta=10^{-4}$;
for FRF-XXX model $J=1$ and $h^{\alpha}_{i}\in[-0.5,0.5]$; and for
the MFI model $J=1$, $g^{x}=0.8$, $g^{z}=0.9$, and $d_{L/2}=0.1$.
For (a), the Hilbert space dimension is $N=2^{12}$ with $50$ realizations,
while for (b-d) it is $N=2^{11}$ together with an average over $500$
realizations. Further details of the analysis of this figure and the
numerical implementation are presented in Sec.\ref{sec:Numerical-Calculation}
of the Supplemental material~\citep{SM}.}
\end{figure}

\textit{Discussions and Conclusion.}--- For qubit systems where $N=2^{L}$,
$P_{\alpha}$, $P_{\beta}$ are Pauli strings satisfying $\mathrm{Tr}P_{\alpha}=0,\,\mathrm{Tr}(P_{\alpha}P_{\beta})=N\delta_{\alpha\beta}$,
in the case of random matrices as well randomized protocols, the late-time
plateau of $\langle\!\langle\overline{Q}^{\mathrm{ctn}}_{\alpha\alpha}(T)\rangle\!\rangle/T^{2}$
scales as
\begin{equation}
\lim_{T\to\infty}{\displaystyle \langle\!\langle\overline{Q}^{\mathrm{ctn}}_{\alpha\alpha}(T)\rangle\!\rangle/T^{2}\sim1/N=1/2^{L}},
\end{equation}
independent of the size of the support of the of coupling operators
$P_{\alpha}$, see numerical checks in the Supplemental Material~\citep{SM}.
It was shown in Ref.~\citep{shi2024universal} that due to the Lieb-Robinson
bound, at short times QFI scales linearly with $L$ for local Hamiltonians
and initial states with low entanglement. Furthermore, It was shown\textcolor{blue}{{}
}that in random quantum circuit, the typical QFI also scales linearly
with $L$ in the limit where the local Hilbert space dimension is
much larger than the discrete time~\citep{bergman2026asymptotic}.
Here, we show that typical stationary value of QFI decays exponentially
with the system size. 

In the case of finite $\bm{\theta}$, it can be readily seen that
the average QFI probes the SFF of $H_{\bm{\theta}}$ rather than $H$
in both, the randomized and deterministic coupling protocols. For
the case where $H$ is chaotic and $P_{\alpha}$ is integrable or
non-interacting, due to the connection between the average QFI and
SFF, the average QFI for different values of $\bm{\theta}$ also encodes
the information about chaos-to-integrability transitions, which is
left for future study.

The Haar ensemble of initial states and coupling operators can be
relaxed to the ensemble of initial states which forms a global state
$2$-design or global unitary $2$-design. Recently, there are numerous
proposals to implement higher-order unitary and state designs with
chaotic Hamiltonians~\citep{zhou2025realizing,sun2026unitary,hou2026statekdesigns}.
Nevertheless, our results only requires the ensemble of states and
unitary $2$-designs, which can be easily implemented. For fixed initial
state, we expect that the QFI can probe the higher-order spectral
correlations~\citep{liu2018spectral,oliviero2021randommatrix,cotler2020spectral}. 

To summarize, we reveal a fundamental connection between the average
QFI, related to quantum noise and the SFF, the diagnostics of quantum
chaos. Refs.~\citep{vsuntajs2020quantum,schiulaz2019thouless} challenge
the existence of many-body localization~\citep{sierant2025manybody,nandkishore2015manybody}
by investigating the characteristics of the SFF in the ergodic and
localized phase. Since the measurement of QFI have been realized in
several experimental platforms~\citep{strobel2014fisherinformation,zhang2026gapclosing,marciniak2022optimal},
we believe our proposal may provide a feasible way to measure the
SFF via QFI. Many intriguing questions are open in the future, including
establishing the connection between QFI or GT in open quantum systems,
investigate the statistics of the random QFI or GT~\citep{berry2020quantum,berry2020geometric},
etc. 

\textit{Acknowledgement.}--- It is a pleasure to thank Yunxiang Liao
and Lea F. Santos for helpful discussions. Codes used for the current
study are available from the authors upon requests. We acknowledge
support by the Natural Science Foundation of China (NSFC) under Grant
No.~12675029, Zhejiang University Start-up Grant, Zhejiang Key Laboratory
of R\&D and Application of Cutting-edge Scientific Instruments. During
the preparation of this manuscript, ChatGPT 6 was used to help write
the codes, suggest references, and inspire proof strategies. The main
ideas, steps and writing were all done without the use of AI and the
authors take full responsibility for the contents of this manuscript.

\let\oldaddcontentsline\addcontentsline     
\renewcommand{\addcontentsline}[3]{}         

\bibliographystyle{apsrev4-1}
\bibliography{ChaoticMetrology}

\clearpage\newpage\setcounter{equation}{0} \setcounter{section}{0}
\setcounter{subsection}{0} 
\global\long\def\theequation{S\arabic{equation}}%
\onecolumngrid \setcounter{enumiv}{0} 

\setcounter{equation}{0} \setcounter{section}{0} \setcounter{subsection}{0} \renewcommand{\theequation}{S\arabic{equation}} \onecolumngrid \setcounter{enumiv}{0} \setcounter{figure}{0} \renewcommand{\thefigure}{S\arabic{figure}}
\begin{center}
{\large\textbf{Supplemental Material}}{\large\par}
\par\end{center}

\let\addcontentsline\oldaddcontentsline     

\tableofcontents{}

\addtocontents{toc}{\protect\thispagestyle{empty}}
\pagenumbering{gobble}

\section{Multi-time Isospectral Twirling\label{sec:Isospectral-twirling}}

We consider random isospectral Hamiltonian ensemble with $H=U\Lambda U^{\dagger}$,
where the unitary matrix $U$ and the diagonal matrix $\Lambda=\mathrm{diag}(\lambda_{1},\,\lambda_{2},\cdots\lambda_{N})$
are statistically independent and sampled from the Haar random measure
and the probability measure $\mu(\Lambda)$, respectively. For the
average of GT over the isospectral Hamiltonian ensemble , i.e., 
\begin{equation}
\langle Q_{\alpha\beta}\rangle\equiv\mathbb{E}_{H\sim\mathrm{ISO}}[Q_{\alpha\beta}(H,\,\ket{\psi_{0}})],
\end{equation}
the calculation involves the notion of isospectral twirling introduced
in Ref.~\citep{oliviero2021randommatrix} as we shall see later.
Nevertheless, the isospectral twirling introduced in Ref.~\citep{oliviero2021randommatrix}
only involves a single time while the isospectral twirling here involves
multiple times. To be more specific, we define

\begin{equation}
\mathscr{T}^{(k)}_{\mathrm{ISO}}(X_{1},\,t_{1};X_{2},t_{2};\cdots X_{k},t_{k})\equiv\mathbb{E}_{H\sim\mathrm{ISO}}[\otimes^{k}_{\alpha=1}\mathrm{e}^{\mathrm{i}Ht_{\alpha}}X_{\alpha}\mathrm{e}^{\mathrm{i}Ht_{\alpha}}].
\end{equation}
As we shall see, one can relate $\mathscr{T}^{(k)}_{\mathrm{ISO}}(X_{1},\,t_{1};X_{2},t_{2};\cdots X_{k},t_{k})$
to the Haar twirling of operator through the following lemma:
\begin{lem}
\textup{It can be shown that 
\begin{equation}
\mathscr{T}^{(k)}_{\mathrm{ISO}}(X_{1},\,t_{1};X_{2},t_{2};\cdots X_{k},t_{k})=\mathbb{E}_{\Lambda\sim P(\Lambda)}\left\{ \mathrm{Tr}_{(1,2,\cdots k)}\left[\otimes^{k}_{\alpha=1}X_{\alpha}\mathscr{T}^{(2k)}_{\mathrm{Haar}}\left((\otimes^{k}_{\beta=1}\mathrm{e}^{-\mathrm{i}\Lambda t_{k}})\otimes\otimes^{k}_{\beta=1}\mathrm{e}^{\mathrm{i}\Lambda t_{k}})\right)\prod^{k}_{\delta=1}\mathbb{F}_{\delta,\,k+\delta}\right]\right\} ,
\end{equation}
where }$V=U\Lambda U^{\dagger}$\textup{ and
\begin{equation}
\mathscr{T}^{(2k)}_{\mathrm{Haar}}(X)\equiv\mathbb{E}_{U\sim\mathrm{Haar}}[U^{\otimes2k}XU^{\dagger\otimes2k}]
\end{equation}
is the twirling operation $\mathbb{F}_{j,\,k}$ is the swap operator
between the $j$-th and $k$-th Hilbert spaces and $\mathrm{Tr}_{(1,2,\cdots,\,k)}$
is the partial trace over the first half of the Hilbert space. }
\end{lem}
\begin{proof}
The proof is straightforward: Using the partial swap trick $AB=\mathrm{Tr}_{2}(A\otimes B\mathbb{F})$,
we find 
\begin{align}
\otimes^{k}_{\alpha=1}\left(\mathrm{e}^{\mathrm{i}Ht_{\alpha}}X_{\alpha}\mathrm{e}^{-\mathrm{i}Ht_{\alpha}}\right) & =\otimes^{k}_{\alpha=1}\left(U\mathrm{e}^{\mathrm{i}\Lambda t_{\alpha}}U^{\dagger}X_{\alpha}U\mathrm{e}^{-\mathrm{i}\Lambda t_{\alpha}}U^{\dagger}\right)\nonumber \\
 & =\otimes^{k}_{\alpha=1}\left(\mathrm{Tr}_{2\alpha}(\mathbb{I}\otimes X_{\alpha})U^{\otimes2}(\mathrm{e}^{\mathrm{i}\Lambda t_{\alpha}}\otimes\mathrm{e}^{-\mathrm{i}\Lambda t_{\alpha}})U^{\dagger\otimes2}\mathbb{F}_{2\alpha-1,2\alpha}\right)\nonumber \\
 & =\mathrm{Tr}_{2,4\cdots,2k}\left[(\mathbb{I}\otimes X_{1}\otimes\cdots\otimes\mathbb{I}\otimes X_{k})U^{\otimes2k}(\mathrm{e}^{\mathrm{i}\Lambda t_{1}}\otimes\mathrm{e}^{-\mathrm{i}\Lambda t_{1}}\cdots\mathrm{e}^{\mathrm{i}\Lambda t_{k}}\otimes\mathrm{e}^{-\mathrm{i}\Lambda t_{k}})U^{\dagger\otimes2k}(\mathbb{F}_{12}\otimes\mathbb{F}_{34}\cdots\mathbb{F}_{2k-1,2k})\right].
\end{align}
Upon permuting the order of Hilbert spaces, it can be readily found
\begin{equation}
\otimes^{k}_{\alpha=1}\left(\mathrm{e}^{\mathrm{i}Ht_{\alpha}}X_{\alpha}\mathrm{e}^{-\mathrm{i}Ht_{\alpha}}\right)=\mathrm{Tr}_{(1,2,\cdots k)}\left[\left((\otimes^{k}_{\alpha=1}X_{\alpha})U^{\otimes2k}(\otimes^{k}_{\beta=1}\mathrm{e}^{-\mathrm{i}\Lambda t_{\alpha}}\otimes^{k}_{\gamma=1}\mathrm{e}^{\mathrm{i}\Lambda t_{\alpha}})U^{\dagger\otimes2k}\prod^{k}_{\delta=1}\mathbb{F}_{\delta,\,k+\delta}\right)\right],
\end{equation}
which concludes the proof.
\end{proof}

Let us first note, for the continuous protocol, the average of GT
over the isospectral ensemble for a given initial state is 
\begin{align}
\langle\!\langle Q^{\mathrm{ctn}}_{\alpha\beta}\rangle & \!\rangle=4\int^{T}_{0}ds\int^{T}_{0}d\tau\mathrm{Tr}\left(\rho_{0}\mathbb{E}_{H\sim\mathrm{ISO}}[e^{\mathrm{i}Hs}P_{\alpha}e^{-\mathrm{i}Hs}e^{\mathrm{i}H\tau}P_{\beta}e^{-\mathrm{i}H\tau}]\right)\nonumber \\
 & -4\int^{T}_{0}ds\int^{T}_{0}d\tau\mathbb{E}_{H\sim\mathrm{ISO}}\left[\mathrm{Tr}\left(\rho_{0}e^{\mathrm{i}Hs}P_{\alpha}e^{-\mathrm{i}Hs}\right)\mathrm{Tr}\left(\rho_{0}e^{\mathrm{i}H\tau}P_{\beta}e^{-\mathrm{i}H\tau}\right)\right],
\end{align}
while for the quench protocol, the corresponding result is
\begin{equation}
\langle\!\langle Q^{\mathrm{qch}}_{\alpha\beta}\rangle\!\rangle=4T^{2}\mathrm{Tr}(\rho_{0}\mathbb{E}_{H\sim\mathrm{ISO}}[e^{\mathrm{i}HT}P_{\alpha}e^{-\mathrm{i}HT}e^{\mathrm{i}HT}P_{\beta}e^{-\mathrm{i}HT}])-\mathbb{E}_{H\sim\mathrm{GUE}}[\mathrm{Tr}(\rho_{0}e^{\mathrm{i}HT}P_{\alpha}e^{-\mathrm{i}HT})\mathrm{Tr}(\rho_{0}e^{\mathrm{i}HT}P_{\beta}e^{-\mathrm{i}HT})].
\end{equation}
Using the partial swap trick and $\mathrm{Tr}A\mathrm{Tr}B=\mathrm{Tr}(A\otimes B)$,
we obtain 
\begin{equation}
\langle\!\langle Q^{\mathrm{ctn}}_{\alpha\beta}\rangle\!\rangle=4\int^{T}_{0}ds\int^{T}_{0}d\tau\left[\mathrm{Tr}\left(\rho_{0}\mathrm{Tr}_{2}\left(\mathscr{T}^{(2)}_{\mathrm{ISO}}(P_{\alpha},s;P_{\beta},\tau)\mathbb{F}_{12}\right)\right)-\mathrm{Tr}\left(\rho_{0}\otimes\rho_{0}\mathscr{T}^{(2)}_{\mathrm{ISO}}(P_{\alpha},s;P_{\beta},\tau)\right)\right].\label{eq:Q-ctn-GUE}
\end{equation}
Similarly for the quench protocol, we find 
\begin{equation}
\langle\!\langle Q^{\mathrm{qch}}_{\alpha\beta}\rangle\!\rangle=4T^{2}\left[\mathrm{Tr}\left(\rho_{0}\mathrm{Tr}_{2}\left(\mathscr{T}^{(2)}_{\mathrm{ISO}}(P_{\alpha},t;P_{\beta},t)\mathbb{F}_{12}\right)\right)-\mathrm{Tr}\left(\rho_{0}\otimes\rho_{0}\mathscr{T}^{(2)}_{\mathrm{ISO}}(P_{\alpha},t;P_{\beta},t)\right)\right],
\end{equation}
where 
\begin{equation}
\mathrm{Tr}_{2}\left(\mathscr{T}^{(2)}_{\mathrm{ISO}}(X_{1},t;X_{2},t)\mathbb{F}_{12}\right)=\mathscr{T}^{(1)}_{\mathrm{ISO}}(X_{1}X_{2},t).\label{eq:T2-T1-GUE-id}
\end{equation}

For $k=1$, using Schur-Weyl duality~\citep{roberts2017chaosand,mele2024introduction},
we can find that 
\begin{equation}
\mathscr{T}^{(1)}_{\mathrm{ISO}}(X,t)=\mathbb{E}_{\Lambda\sim P(\Lambda)}\left\{ \mathrm{Tr}_{1}\left[X\mathscr{T}^{(2)}_{\mathrm{Haar}}\left(\mathrm{e}^{-\mathrm{i}\Lambda t}\otimes\mathrm{e}^{\mathrm{i}\Lambda t})\right)\mathbb{F}_{12}\right]\right\} =p(t)X+[1-p(t)]\frac{\mathrm{Tr}X}{N}\mathbb{I},\label{eq:T1-GUE}
\end{equation}
where 
\begin{equation}
p(t)\equiv\frac{\tilde{K}(t)-1}{N^{2}-1},
\end{equation}
and $\tilde{K}(t)$ is the unnormalized SFF defined as
\begin{equation}
\tilde{K}(t)\equiv\langle\sum_{kl}\mathrm{e}^{\mathrm{i}(\lambda_{k}-\lambda_{l})}\rangle_{\Lambda}=\langle\!\langle|\mathrm{Tr}\mathrm{e}^{\mathrm{i}Ht}|^{2}\rangle\!\rangle.\label{eq:K-tilde-def}
\end{equation}
However, the continuous and quench protocols concern with the case
$k=2$ where one needs to evaluate 2-copy twirling of the dynamics
generated by GUE, i..e, $\mathscr{T}^{(2)}_{\mathrm{GUE}}$, which
amounts to the evaluation of $4$-copy twirling of the Haar random
gates, i.e.,

\begin{equation}
\mathscr{T}^{(4)}_{\mathrm{Haar}}\left((\otimes^{2}_{\beta=1}\mathrm{e}^{-\mathrm{i}\Lambda t_{k}})\otimes(\otimes^{2}_{\beta=1}\mathrm{e}^{\mathrm{i}\Lambda t_{k}})\right)=\sum_{\pi,\sigma\in S_{4}}\mathrm{Wg}_{4}(\pi^{-1}\sigma)\mathrm{Tr}\left(T^{\dagger}_{4}(\sigma)(\otimes^{2}_{\beta=1}\mathrm{e}^{-\mathrm{i}\Lambda t_{k}})\otimes(\otimes^{2}_{\beta=1}\mathrm{e}^{\mathrm{i}\Lambda t_{k}})\right)T_{4}(\pi),\label{eq:T4-Haar}
\end{equation}
where $T_{k}(\pi)$ is the representation of the permutation group
$S_{N}$ in $(\mathbb{C}^{N})^{\otimes k}$. The exact expression
of Eq.~(\ref{eq:T4-Haar}) can be obtained via Schur-Weyl duality,
which involves a cumbersome expression containing $(4!)^{2}$ terms.

\section{\label{sec:QFI-closed-form}Average of QFI over Haar ensemble of
initial states and Isospectral ensemble}

In this section, we further perform the average over the Haar ensemble
of initial states, apart from the average over the isospectral ensemble.
Recall 
\begin{align}
\overline{Q}_{\alpha\beta} & \equiv\mathbb{E}_{\ket{\psi_{0}}\sim\mathrm{Haar}}[Q_{\alpha\beta}(H,\,\ket{\psi_{0}}),\\
\langle\!\langle\overline{Q}_{\alpha\beta}\rangle\!\rangle & \equiv\mathbb{E}_{H\sim\mathrm{ISO}}\mathbb{E}_{\ket{\psi_{0}}\sim\mathrm{Haar}}[Q_{\alpha\beta}(H,\,\ket{\psi_{0}})].
\end{align}
 Since 
\begin{equation}
\mathbb{E}_{\ket{\psi_{0}}\sim\mathrm{Haar}}\ket{\psi_{0}}\bra{\psi_{0}}=\frac{\mathbb{I}}{N},\,\mathbb{E}_{\ket{\psi_{0}}\sim\mathrm{Haar}}\ket{\psi_{0}}\bra{\psi_{0}}^{\otimes2}=\frac{\mathbb{I}+\mathbb{F}}{N(N+1)},
\end{equation}
we find
\begin{equation}
\mathbb{E}_{\ket{\psi_{0}}\sim\mathrm{Haar}}(\mathrm{Tr}\ket{\psi_{0}}\bra{\psi_{0}}AB)=\frac{1}{N}\mathrm{Tr}(AB),\,\mathbb{E}_{\ket{\psi_{0}}\sim\mathrm{Haar}}\mathrm{Tr}[\ket{\psi_{0}}\bra{\psi_{0}}^{\otimes2}A\otimes B]=\frac{\mathrm{Tr}A\mathrm{Tr}B+\mathrm{Tr}(AB)}{N(N+1)},
\end{equation}
and
\begin{equation}
\overline{Q}_{\alpha\beta}\equiv\mathbb{E}_{\ket{\psi_{0}}\sim\mathrm{Haar}}[Q_{\alpha\beta}(\ket{\psi_{0}})]=\frac{4[N\mathrm{Tr}(G_{\alpha}G_{\beta})-\mathrm{Tr}G_{\alpha}\mathrm{Tr}G_{\beta}]}{N(N+1)}.\label{eq:Q-bar}
\end{equation}
For the continuous protocol, we find 
\begin{equation}
\langle\!\langle\overline{Q}^{\mathrm{ctn}}_{\alpha\beta}(T)\rangle\!\rangle=\frac{4\left(N\int^{T}_{0}ds\int^{T}_{0}d\tau\mathrm{Tr}[\mathscr{T}^{(1)}_{\mathrm{ISO}}(P_{\alpha},s-\tau)P_{\beta}]-T^{2}\mathrm{Tr}P_{\alpha}\mathrm{Tr}P_{\beta}\right)}{N(N+1)}.\label{eq:Q-ctn-ISO}
\end{equation}
Using Eq.~(\ref{eq:T1-GUE}), we find that 
\begin{align}
\mathrm{Tr}[\mathscr{T}^{(1)}_{\mathrm{ISO}}(P_{\alpha},s-\tau)P_{\beta}] & =\frac{\tilde{K}(s-\tau)-1}{N^{2}-1}\mathrm{Tr}(P_{\alpha}P_{\beta})-\frac{\tilde{K}(s-\tau)-N^{2}}{N(N^{2}-1)}\mathrm{Tr}P_{\alpha}\mathrm{Tr}P_{\beta}\nonumber \\
 & =\frac{N\mathrm{Tr}(P_{\alpha}P_{\beta})-\mathrm{Tr}P_{\alpha}\mathrm{Tr}P_{\beta}}{(N-1)N(N+1)}\tilde{K}(s-\tau)-\frac{\mathrm{Tr}(P_{\alpha}P_{\beta})-N\mathrm{Tr}P_{\alpha}\mathrm{Tr}P_{\beta}}{(N-1)(N+1)},
\end{align}
and therefore
\begin{equation}
\langle\!\langle\mathrm{Tr}(G_{\alpha}G_{\beta})\rangle\!\rangle=\frac{N\mathrm{Tr}(P_{\alpha}P_{\beta})-\mathrm{Tr}P_{\alpha}\mathrm{Tr}P_{\beta}}{(N-1)N(N+1)}\int^{T}_{0}ds\int^{T}_{0}d\tau\tilde{K}(s-\tau)-\frac{\mathrm{Tr}(P_{\alpha}P_{\beta})-N\mathrm{Tr}P_{\alpha}\mathrm{Tr}P_{\beta}}{(N-1)(N+1)}T^{2},\label{eq:Tr-GaGb}
\end{equation}
where $\tilde{K}(t)$ is defined in Eq.\ref{eq:K-tilde-def}. Thus
we obtain 
\begin{equation}
\langle\!\langle\bar{Q}^{\mathrm{ctn}}_{\alpha\beta}(T)\rangle\!\rangle=\frac{4\left(N\mathrm{Tr}(P_{\alpha}P_{\beta})-\mathrm{Tr}P_{\alpha}\mathrm{Tr}P_{\beta}\right)}{(N-1)N(N+1)^{2}}\left[2\int^{T}_{0}(T-t)\tilde{K}(t)dt-T^{2}\right].
\end{equation}
Using the identity 
\begin{equation}
\int^{T}_{0}ds\int^{T}_{0}d\tau\tilde{K}(s-\tau)=2\int^{T}_{0}(T-t)\tilde{K}(t)dt,\label{eq:double-int-id}
\end{equation}
and $K(t)=\tilde{K}(t)/N^{2}$, we arrive at Eq.~(\ref{eq:Q-ctn})
in the main text. Eq.~(\ref{eq:Q-qch}) can be obtained in a similar
way. 

\section{\label{sec:Closed-form-expression-GUE}Closed-form expression for
the average QFIM over Haar random initial states and GUE}

In this section, we present the detailed calculation of the closed
form expression for the average QFIM/GT when the Hamiltonian ensemble
is GUE. The closed-form expression for SFF for GUE is\ \citep{cotler2020spectral}
\begin{equation}
K(t)=\begin{cases}
\frac{1}{t^{2}}J^{2}_{1}(2t)+\frac{t}{2N^{2}} & t\leq2N\\
\frac{1}{t^{2}}J^{2}_{1}(2t)+\frac{1}{N} & t\ge2N
\end{cases},\label{eq:GUE_K(t)}
\end{equation}
where $J_{\nu}(t)$ are the first-kind Bessel function of order $\nu$.
The smoothed integral SFF $I(T)$ defined as 
\begin{equation}
I(T)=\int^{1}_{0}(1-u)K(uT)du=\frac{1}{T^{2}}\int^{T}_{0}(T-t)K(t)dt,
\end{equation}
consists of three parts. The first part appears for any time $T$
\begin{equation}
{\displaystyle I_{1}(T)}=\frac{1}{T^{2}}\int^{T}_{0}dt{\displaystyle (T-t)\frac{1}{t^{2}}J^{2}_{1}(2t)}=\frac{1}{T^{2}}\left[\mathcal{J}(T)-\frac{1}{2}\right],
\end{equation}
where 
\begin{equation}
\mathcal{J}(t)\equiv{\displaystyle \left[\frac{8}{3}t^{2}+\frac{1}{2}\right]}{\displaystyle {\displaystyle {\displaystyle \left[J^{2}_{0}(2t)+J^{2}_{1}(2t)\right]}-\frac{1}{3}J^{2}_{1}(2t)-\frac{4}{3}tJ_{0}(2t)J_{1}(2t)}}.
\end{equation}
The second part is applicable when $T\leq2N$,
\begin{equation}
{\displaystyle {\displaystyle I_{2}}}=\frac{1}{T^{2}}\int^{T}_{0}dt{\displaystyle (T-t)\frac{t}{2N^{2}}={\displaystyle \frac{T}{12N^{2}}.}}
\end{equation}
The last part is applicable when $T\ge2N$, i.e.,
\begin{equation}
{\displaystyle {\displaystyle I_{3}}}=\frac{1}{T^{2}}\left[\int^{2N}_{0}dt{\displaystyle (T-t)\frac{t}{2N^{2}}+\int^{T}_{2N}dt{\displaystyle (T-t)\frac{1}{N}}}\right]=\frac{1}{2N}-\frac{1}{T}+\frac{2N}{3T^{2}}.
\end{equation}
So, the integral of the SFF is 

\begin{equation}
I(T)=\begin{cases}
\frac{\mathcal{J}(T)}{T^{2}}+\frac{T}{12N^{2}}-\frac{1}{2T^{2}} & T\leq2N\\
\frac{\mathcal{J}(T)}{T^{2}}+\frac{1}{2N}-\frac{1}{T}+\left(\frac{2N}{3}-\frac{1}{2}\right)\frac{1}{T^{2}} & T\ge2N
\end{cases}.\label{eq:GUE-Integral-SFF}
\end{equation}
Since the Bessel functions behave as $J_{\nu}(t)\sim{\displaystyle \sqrt{{\displaystyle \frac{2}{\pi t}}}}\cos{\displaystyle \left(t-\frac{1}{2}\nu\pi-\frac{1}{4}\pi\right)}$
as $t\to\infty$, as $T\to\infty$

\begin{align}
J_{0}(2T)\sim{\displaystyle {\displaystyle \frac{1}{{\displaystyle \sqrt{\pi T}}}}}\cos{\displaystyle \left(2T-\frac{1}{4}\pi\right)},\, & J_{1}(2T)\sim{\displaystyle {\displaystyle \frac{1}{{\displaystyle \sqrt{\pi T}}}}}\sin{\displaystyle \left(2T-\frac{1}{4}\pi\right)}\\
J_{0}(2T)J_{1}(2T)\sim{\displaystyle \frac{1}{\pi T}}\cos\mbox{ \ensuremath{{\displaystyle \left(2T-\frac{1}{4}\pi\right)}}}\sin{\displaystyle \left(2T-\frac{1}{4}\pi\right)},\, & J^{2}_{0}(2T)+J^{2}_{1}(2T)\sim\frac{1}{\pi T}
\end{align}
Therefore, ${\displaystyle \lim_{T\to\infty}\frac{\mathcal{J}(T)}{T^{2}}}=0$,
\begin{equation}
\lim_{T\to\infty}I(T)={\displaystyle \frac{1}{2N}},\quad\lim_{T\to\infty}\langle\!\langle\bar{Q}^{\mathrm{ctn}}_{\alpha\beta}(T)\rangle\!\rangle=\frac{C_{\alpha\beta}T^{2}}{N+1}.
\end{equation}
On the other hand, in the short-time limit using Taylor series up
to ${\displaystyle O{\displaystyle (T^{2})}}$ we have $J_{0}(2T)\sim1-T^{2}$,
$J_{1}(2T)\sim T-{\displaystyle \frac{1}{2}}T^{3}$, $J^{2}_{0}(2T)\sim1-2T^{2}$,
$J^{2}_{1}(2T)\sim T^{2}$, $J_{0}(2T)J_{1}(2T)\sim T-{\displaystyle \frac{1}{2}}T^{2}$,
and $J^{2}_{0}(2T)+J^{2}_{1}(2T)\sim1-T^{2}$, we find

\begin{equation}
{\displaystyle \frac{\mathcal{J}(T)}{T^{2}}}\sim{\displaystyle \frac{1}{2}-\frac{8}{3}T^{2}+\frac{1}{2T^{2}}}+\frac{2T}{3},
\end{equation}
and 
\begin{equation}
I(T)={\displaystyle \frac{\mathcal{J}(T)}{T^{2}}}{\displaystyle +\frac{T}{12N^{2}}-\frac{1}{2T^{2}}}\sim{\displaystyle \frac{1}{2}-\frac{8}{3}T^{2}}+\frac{2T}{3}+\frac{T}{12N^{2}}.
\end{equation}
Therefore we have 
\begin{equation}
{\displaystyle \lim_{T\to0}I(T)}={\displaystyle \frac{{\displaystyle 1}}{2}},\;\lim_{T\to0}\langle\!\langle\bar{Q}^{\mathrm{ctn}}_{\alpha\beta}(T)\rangle\!\rangle=4C_{\alpha\beta}T^{2}.
\end{equation}

\section{Average QFI over Haar random initial states and realistic Hamiltonian
ensemble}

In this section, we give the detailed derivations of generalizing
the average QFI from the random matrix predictions to realistic Hamiltonian
ensemble. Throughout this section, we consider\textcolor{black}{{} the
continuous protocol with }the average over the ensemble of initial
Haar random states. According to Eq.~(\ref{eq:Q-bar}), the only
term that depends on the Hamiltonian ensemble is
\begin{equation}
\mathrm{Tr}(G^{\mathrm{ctn}}_{\alpha}G^{\mathrm{ctn}}_{\beta})=\int^{T}_{0}ds\int^{T}_{0}d\tau\mathrm{Tr}\left[\mathrm{e}^{\mathrm{i}H(s-\tau)}P_{\alpha}\mathrm{e}^{-\mathrm{i}H(s-\tau)}P_{\beta}\right].\label{eq:Tr-Gctn-Gctn}
\end{equation}

\subsection{2-copy twirling of the coupling operators}

It is clear from Eq.\ (\ref{eq:Tr-Gctn-Gctn}) that the 1-copy twirling
of the coupling operators, where the Haar random unitary acting on
$P_{\alpha}$ and $P_{\beta}$ are statistically independent, leads
to the same result as in the quench protocol, which is independent
of the Hamiltonian. Therefore we consider 2-copy twirling, where the
unitary acting on $P_{\alpha}$ and $P_{\beta}$ are identical. After
the average over the Haar ensemble and the Hamiltonian ensemble, we
obtain
\begin{align}
\langle\!\langle\mathrm{Tr}(\widetilde{G^{\mathrm{ctn}}_{\alpha}G^{\mathrm{ctn}}_{\beta}}) & \rangle\!\rangle=\int d\mu(H)\int^{T}_{0}ds\int^{T}_{0}d\tau\mathrm{Tr}\mathbb{E}_{V\sim\mathrm{Haar}}\left[\mathrm{e}^{\mathrm{i}H(s-\tau)}VP_{\alpha}V^{\dagger}\mathrm{e}^{-\mathrm{i}H(s-\tau)}VP_{\beta}V^{\dagger}\right]\nonumber \\
 & =\int^{T}_{0}ds\int^{T}_{0}d\tau\mathrm{Tr}\int d\mu(H)\mathbb{E}_{V\sim\mathrm{Haar}}\left[\mathrm{e}^{\mathrm{i}V^{\dagger}HV(s-\tau)}P_{\alpha}\mathrm{e}^{-\mathrm{i}V^{\dagger}H(s-\tau)V}P_{\beta}\right],
\end{align}
where $\langle\!\langle\bullet\rangle\!\rangle$ denotes the average
over the Hamiltonian ensemble, which follows a generic probability
measure $d\mu_{H}$ and is not Haar random, and $\langle\bullet\rangle_{P}$
denotes the Haar average over the coupling operators. It is crucial
to observe that the ensemble $\{H^{\prime}=V^{\dagger}HV:V\sim\mathrm{Haar},H\sim d\mu(H)\}$
is an isospectral ensemble since for given $H=U\Lambda U^{\dagger}$,
thanks to the invariance of Haar measure, we find 
\begin{equation}
\int d\mu(H)\mathbb{E}_{V\sim\mathrm{Haar}}f(H^{\prime})=\int d\mu(U\Lambda U^{\dagger})\mathbb{E}_{V\sim\mathrm{Haar}}f(V^{\dagger}U\Lambda VU^{\dagger})=\int d\mu(\Lambda)\mathbb{E}_{V\sim\mathrm{Haar}}f(V^{\dagger}\Lambda V),
\end{equation}
where $d\mu(\Lambda)\equiv\int_{U}d\mu(U\Lambda U^{\dagger})$ is
the marginalized probability measure upon integrating over the unitary
matrix $U$ that diagonalizes the Hamiltonian $H$. Note that in above
derivation neither do we assume the statistical independence between
$U$ and $\Lambda$ nor assume the permutation invariance of the probability
measure $d\mu(\Lambda)$. Therefore, conclude that 
\begin{equation}
\mathrm{Tr}\int d\mu(H)\mathbb{E}_{V\sim\mathrm{Haar}}\left[\mathrm{e}^{\mathrm{i}V^{\dagger}HV(s-\tau)}P_{\alpha}\mathrm{e}^{-\mathrm{i}V^{\dagger}H(s-\tau)V}P_{\beta}\right]=\mathrm{Tr}[\mathscr{T}^{(1)}_{\mathrm{ISO}}(P_{\alpha},s-\tau)P_{\beta}].
\end{equation}
One can also check explicitly by noting that
\begin{equation}
\mathrm{Tr}\left[\mathrm{e}^{\mathrm{i}H(s-\tau)}VP_{\alpha}V^{\dagger}\mathrm{e}^{-\mathrm{i}H(s-\tau)}VP_{\beta}V^{\dagger}\right]=\mathrm{Tr}\left[V^{\otimes2}(P_{\alpha}\otimes P_{\beta})V^{\dagger\otimes2}(\mathrm{e}^{-\mathrm{i}H(s-\tau)}\otimes\mathbb{I})(\mathbb{I}\otimes\mathrm{e}^{\mathrm{i}H(s-\tau)})\mathbb{F}\right],
\end{equation}
and 
\begin{align}
 & \mathrm{Tr}\mathbb{E}_{V\sim\mathrm{Haar}}\left[V^{\otimes2}(P_{\alpha}\otimes P_{\beta})V^{\dagger\otimes2}(\mathrm{e}^{-\mathrm{i}H(s-\tau)}\otimes\mathbb{I})(\mathbb{I}\otimes\mathrm{e}^{\mathrm{i}H(s-\tau)})\mathbb{F}\right]\nonumber \\
= & \frac{\mathrm{Tr}P_{\alpha}\mathrm{Tr}P_{\beta}-N^{-1}\mathrm{Tr}(P_{\alpha}P_{\beta})}{N^{2}-1}N+\frac{\mathrm{Tr}(P_{\alpha}P_{\beta})-N^{-1}\mathrm{Tr}P_{\alpha}\mathrm{Tr}P_{\beta}}{N^{2}-1}\big|\mathrm{Tr}\mathrm{e}^{\mathrm{i}H(s-\tau)}|^{2}\nonumber \\
= & \frac{N\mathrm{Tr}(P_{\alpha}P_{\beta})-\mathrm{Tr}P_{\alpha}\mathrm{Tr}P_{\beta}}{(N-1)N(N+1)}\big|\mathrm{Tr}\mathrm{e}^{\mathrm{i}H(s-\tau)}|^{2}-\frac{\mathrm{Tr}(P_{\alpha}P_{\beta})-N\mathrm{Tr}P_{\alpha}\mathrm{Tr}P_{\beta}}{N^{2}-1}.
\end{align}
Clearly, upon the average over the ensemble of random Hamiltonians
and the integration over $s$ and $\tau$ from $[0,T]$, one can readily
obtain Eq.\ (\ref{eq:Tr-GaGb})

As such, we conclude that upon performing the two-copy twirling of
the coupling operators, for any realistic Hamiltonian, the average
QFI is the same as the one for the isospectral ensemble.

\subsection{Analytic expression based on weak statistical properties of the Hamiltonian
ensemble}

Upon defining $\Gamma^{(\alpha)}_{k\ell}\equiv\braket{\lambda_{k}|P_{\alpha}|\lambda_{\ell}}$,
we can 
\begin{equation}
\mathrm{Tr}(G^{\mathrm{ctn}}_{\alpha}G^{\mathrm{ctn}}_{\beta})=\sum_{k,\,\ell,\,k\neq\ell}\int^{T}_{0}ds\int^{T}_{0}d\tau\mathrm{e}^{\mathrm{i}(s-\tau)(\lambda_{k}-\lambda_{l})}\Gamma^{(\alpha)}_{k\ell}\Gamma^{(\beta)}_{\ell k}+T^{2}\sum_{k}\Gamma^{(\alpha)}_{kk}\Gamma^{(\beta)}_{kk}.
\end{equation}
Using the fact that $\sum_{kl}\Gamma^{(\alpha)}_{k\ell}\Gamma^{(\beta)}_{\ell k}=\mathrm{Tr}(P_{\alpha}P_{\beta})$,
we find
\begin{align}
\langle\!\langle\mathrm{Tr}(G^{\mathrm{ctn}}_{\alpha}G^{\mathrm{ctn}}_{\beta})\rangle\!\rangle & =\int^{T}_{0}ds\int^{T}_{0}d\tau\langle\mathrm{e}^{\mathrm{i}(s-\tau)(\lambda_{k}-\lambda_{\ell})}\rangle\bigg\langle\sum_{k,\,\ell,\,k\neq\ell}\Gamma^{(\alpha)}_{k\ell}\Gamma^{(\beta)}_{\ell k}\bigg\rangle+T^{2}\langle\sum_{k}\Gamma^{(\alpha)}_{kk}\Gamma^{(\beta)}_{kk}\rangle\nonumber \\
 & =\mathrm{Tr}\left(P_{\alpha}P_{\beta}\right)\int^{T}_{0}ds\int^{T}_{0}d\tau\langle\mathrm{e}^{\mathrm{i}(s-\tau)(\lambda_{k}-\lambda_{\ell})}\rangle_{k\neq\ell}-\left(\int^{T}_{0}ds\int^{T}_{0}d\tau\langle\mathrm{e}^{\mathrm{i}(s-\tau)(\lambda_{k}-\lambda_{\ell})}\rangle_{k\neq\ell}-T^{2}\right)\langle\sum_{k}\Gamma^{(\alpha)}_{kk}\Gamma^{(\beta)}_{kk}\rangle.
\end{align}
Since 
\begin{equation}
\tilde{K}(t)=N(N-1)\langle\mathrm{e}^{\mathrm{i}t(\lambda_{k}-\lambda_{\ell})}\rangle_{k\neq\ell}+N,
\end{equation}
 we find 
\begin{equation}
\langle\!\langle\mathrm{Tr}(G^{\mathrm{ctn}}_{\alpha}G^{\mathrm{ctn}}_{\beta})\rangle\!\rangle=\frac{\mathrm{Tr}(P_{\alpha}P_{\beta})-\langle\sum_{k}\Gamma^{(\alpha)}_{kk}\Gamma^{(\beta)}_{kk}\rangle}{N(N-1)}\int^{T}_{0}ds\int^{T}_{0}d\tau\tilde{K}(s-\tau)-\frac{\mathrm{Tr}(P_{\alpha}P_{\beta})-N\langle\sum_{k}\Gamma^{(\alpha)}_{kk}\Gamma^{(\beta)}_{kk}\rangle}{N-1}T^{2}.\label{eq:TrGaGb-general}
\end{equation}

In the case where the eigenstates $\{\ket{\lambda_{k}}\}$ is Haar
random, we find 
\begin{align}
\mathbb{E}_{k,\,l\sim\mathrm{Haar}}[\Gamma^{\alpha}_{k\ell}\Gamma^{\beta}_{\ell k} & ]=\mathbb{E}_{k,\,l\sim\mathrm{Haar}}\left[\braket{\lambda_{k}|P_{\alpha}|\lambda_{l}}\braket{\lambda_{l}|P_{\beta}|\lambda_{k}}\right]\nonumber \\
 & =\begin{cases}
\mathrm{Tr}\left\{ \mathbb{E}_{k\sim\mathrm{Haar}}[(\ket{\lambda_{k}}\bra{\lambda_{k}})^{\otimes2}(P_{\alpha}\otimes P_{\beta})\mathbb{F}_{12}]\right\}  & k=\ell\\
\mathrm{Tr}\left\{ \mathbb{E}_{\ell\sim\mathrm{Haar}}\left[\mathbb{E}_{k\sim\mathrm{Haar}|\ell}[\ket{\lambda_{k}}\bra{\lambda_{k}}]P_{\alpha}P_{\beta}\ket{\lambda_{\ell}}\bra{\lambda_{\ell}}\right]\right\}  & k\neq\ell
\end{cases}.
\end{align}
Using Schur-Weyl duality, one can calculate that
\begin{equation}
\mathbb{E}_{k,\,l\sim\mathrm{Haar}}[\Gamma^{\alpha}_{k\ell}\Gamma^{\beta}_{\ell k}]=\begin{cases}
\frac{\mathrm{Tr}(P_{\alpha}P_{\beta})+\mathrm{Tr}P_{\alpha}\mathrm{Tr}P_{\beta}}{N(N+1)} & k=\ell\\
\frac{N\mathrm{Tr}(P_{\alpha}P_{\beta})-\mathrm{Tr}P_{\alpha}\mathrm{Tr}P_{\beta}}{(N-1)N(N+1)} & k\neq\ell
\end{cases}.\label{eq:C-corr}
\end{equation}
Upon substituting Eq.\ (\ref{eq:TrGaGb-general}) into Eq.\ (\ref{eq:C-corr}),
we arrive at Eq.\ (\ref{eq:Tr-GaGb}). 

In the most generic case, upon substituting Eq.\ (\ref{eq:TrGaGb-general})
into Eq.\ (\ref{eq:Q-bar}), we obtain 
\begin{align}
\langle\!\langle\bar{Q}^{\mathrm{ctn}}_{\alpha\beta}(T)\rangle\!\rangle & =\frac{4[\mathrm{Tr}(P_{\alpha}P_{\beta})-\langle\sum_{k}\Gamma^{(\alpha)}_{kk}\Gamma^{(\beta)}_{kk}\rangle]}{(N-1)N(N+1)}\int^{T}_{0}ds\int^{T}_{0}d\tau\tilde{K}(s-\tau)\nonumber \\
 & -\frac{4[N\mathrm{Tr}(P_{\alpha}P_{\beta})+(N-1)\mathrm{Tr}P_{\alpha}\mathrm{Tr}P_{\beta}-N^{2}\langle\sum_{k}\Gamma^{(\alpha)}_{kk}\Gamma^{(\beta)}_{kk}\rangle]}{(N-1)N(N+1)}T^{2}.
\end{align}
Upon identifying 
\begin{align}
A_{\alpha\beta} & \equiv\frac{1}{N}[\mathrm{Tr}(P_{\alpha}P_{\beta})-\langle\sum_{k}\Gamma^{(\alpha)}_{kk}\Gamma^{(\beta)}_{kk}\rangle],\label{eq:A-def}\\
B_{\alpha\beta} & \equiv\frac{1}{N}[N\mathrm{Tr}(P_{\alpha}P_{\beta})+(N-1)\mathrm{Tr}P_{\alpha}\mathrm{Tr}P_{\beta}-N^{2}\langle\sum_{k}\Gamma^{(\alpha)}_{kk}\Gamma^{(\beta)}_{kk}\rangle],\label{eq:B-def}
\end{align}
we can rewrite 
\begin{equation}
\langle\!\langle\bar{Q}^{\mathrm{ctn}}_{\alpha\beta}(T)\rangle\!\rangle=\frac{4}{N^{2}-1}\left[A_{\alpha\beta}\int^{T}_{0}ds\int^{T}_{0}d\tau\tilde{K}(s-\tau)-B_{\alpha\beta}T^{2}\right],\label{eq:Q-detmn}
\end{equation}
which is identical to Eq.\ (\ref{eq:Q-ctn-phy-H}) in the main text,
upon noting $K(t)=\tilde{K}(t)/N^{2}$.

\section{\label{sec:Orthogonal-Groups}Generalization to Orthogonal groups}

In this section, we discuss isospectral twirling and randomized-coupling
protocol with orthogonal groups, which is relevant for realistic Hamiltonian
ensembles that satisfy GOE level statistics. To this end, we need
to be aware that mathematically the commutant of the tensor representation
of the orthogonal group is now the representation of the Brauer algebra\ \citep{brauer1937onalgebras},
a generalization of the permutation algebra, which contains elements
that represents both permutations and contractions\ \citep{collins2006integration,collins2009onsome,collins2022theweingarten}.
More precisely, 
\begin{equation}
\mathscr{T}^{(2)}_{\mathrm{Haar}}(X)\equiv\mathbb{E}_{O\sim\mathrm{Harr}}[O^{\otimes2}X(O^{T})^{\otimes2}]=c_{\mathbb{I}}(X)\mathbb{I}+c_{\mathbb{F}}(X)\mathbb{F}+c_{\mathbb{E}}(X)\mathbb{E},\label{eq:T2-Haar-GOE}
\end{equation}
where $\mathbb{I}$ and $\mathbb{F}$ are the identity and swap operator
as in the case of unitary groups, $\mathbb{E}=\sum_{ij}\ket{i}\bra{j}\otimes\ket{i}\bra{j}$,
the coefficients can be determined from by solving a set of equations
equation 
\begin{equation}
\begin{bmatrix}\mathrm{Tr}(\mathbb{I}) & \mathrm{Tr}(\mathbb{F}) & \mathrm{Tr}(\mathbb{E})\\
\mathrm{Tr}(\mathbb{F}) & \mathrm{Tr}(\mathbb{F}^{2}) & \mathrm{Tr}(\mathbb{E}\mathbb{F})\\
\mathrm{Tr}(\mathbb{E}) & \mathrm{Tr}(\mathbb{F}\mathbb{E}) & \mathrm{Tr}(\mathbb{E}^{2})
\end{bmatrix}\begin{bmatrix}c_{\mathbb{I}}(X)\\
c_{\mathbb{F}}(X)\\
c_{\mathbb{E}}(X)
\end{bmatrix}=\begin{bmatrix}\mathrm{Tr}X\\
\mathrm{Tr}(\mathbb{F}X)\\
\mathrm{Tr}(\mathbb{E}X)
\end{bmatrix}.
\end{equation}
Now with Eq.\ (\ref{eq:T2-Haar-GOE}), one can easily calculate the
average over the initial Haar random real state vectors, which yields,
for real symmetric operators $A$ and $B$, 
\begin{equation}
\mathbb{E}_{\ket{\psi_{0}}\sim\mathrm{Haar}}(\mathrm{Tr}\ket{\psi_{0}}\bra{\psi_{0}}AB)=\frac{1}{N}\mathrm{Tr}(AB),
\end{equation}
which remains unchanged compared to the case of Haar random complex
vectors. However, for the two-copy twirling, there is a difference
up to a factor of $2$, i.e.,
\begin{equation}
\mathbb{E}_{\ket{\psi_{0}}\sim\mathrm{Haar}}\mathrm{Tr}[\ket{\psi_{0}}\bra{\psi_{0}}^{\otimes2}A\otimes B]=\frac{\mathrm{Tr}A\mathrm{Tr}B+2\mathrm{Tr}(AB)}{N(N+2)}.
\end{equation}
Then we find 
\begin{equation}
\overline{Q}_{\alpha\beta}=\frac{4[N\mathrm{Tr}(G_{\alpha}G_{\beta})-\mathrm{Tr}G_{\alpha}\mathrm{Tr}G_{\beta}]}{N(N+2)}.\label{eq:Q-bar-GOE}
\end{equation}
Therefore for the quench protocol
\begin{equation}
\overline{Q}^{\mathrm{qch}}_{\alpha\beta}(T)=4C^{(\mathrm{R})}_{\alpha\beta}T^{2},
\end{equation}
where 
\begin{equation}
C^{(\mathrm{R})}_{\alpha\beta}\equiv\frac{N\mathrm{Tr}(P_{\alpha}P_{\beta})-\mathrm{Tr}P_{\alpha}\mathrm{Tr}P_{\beta}}{N(N+2)},
\end{equation}
and 
\begin{equation}
\langle\!\langle\overline{Q}^{\mathrm{ctn}}_{\alpha\beta}(T)\rangle\!\rangle=\frac{4\left(N\int^{T}_{0}ds\int^{T}_{0}d\tau\mathrm{Tr}[\mathscr{T}^{(1)}_{\mathrm{ISO}}(P_{\alpha},s-\tau)P_{\beta}]-T^{2}\mathrm{Tr}P_{\alpha}\mathrm{Tr}P_{\beta}\right)}{N(N+2)}.\label{eq:Q-ctn-ISO-GOE}
\end{equation}

For isospectral twirling with the Haar random orthogonal matrices,
we consider $H=U\Lambda U^{\dagger}$ and $U$ is a Haar random orthogonal
matrix, one can show that for a real symmetric operator $X$
\begin{equation}
\mathscr{T}^{(1)}_{\mathrm{ISO}}(X)=\mathbb{E}_{H\sim\mathrm{ISO}}[\mathrm{e}^{\mathrm{i}Ht}X\mathrm{e}^{-\mathrm{i}Ht}]=q(t)X+[1-q(t)]\frac{\mathrm{Tr}X}{N}\mathbb{I},\label{eq:T1-GOE}
\end{equation}
where 
\begin{equation}
q(t)=\frac{\tilde{K}(t)+N-2}{(N-1)(N+2)}.
\end{equation}
Note that we assume $P_{\alpha}$ and $P_{\beta}$ are real and symmetric.
We then substitute Eq.~(\ref{eq:T1-GOE}) into Eq.~(\ref{eq:Q-ctn-ISO-GOE})
and obtain
\begin{align}
\langle\!\langle\overline{Q}^{\mathrm{ctn}}_{\alpha\beta}(T)\rangle\!\rangle & =\frac{4C^{(\mathrm{R})}_{\alpha\beta}}{(N-1)(N+2)}\left[2\int^{T}_{0}(T-t)\tilde{K}(t)dt+(N-2)T^{2}\right]\nonumber \\
 & =\frac{4C^{(\mathrm{R})}_{\alpha\beta}}{(N-1)(N+2)}\left[2N^{2}\int^{T}_{0}(T-t)K(t)dt+(N-2)T^{2}\right].\label{eq:Q-ctn-GOE}
\end{align}
Upon taking the second derivative with respect to $T$, we obtain
\begin{equation}
\frac{d^{2}\langle\!\langle\overline{Q}^{\mathrm{ctn}}_{\alpha\beta}(T)\rangle\!\rangle}{dT^{2}}=\frac{4C^{(\mathrm{R})}_{\alpha\beta}\left[2\tilde{K}(T)+N-2\right]}{(N-1)(N+2)}=\frac{4C^{(\mathrm{R})}_{\alpha\beta}\left[2N^{2}K(T)+N-2\right]}{(N-1)(N+2)}.
\end{equation}
For the deterministic protocol, Eq.~(\ref{eq:TrGaGb-general}) remains
valid. Upon substituting it into Eq.~(\ref{eq:Q-bar-GOE}), we find
\begin{align}
\langle\!\langle\bar{Q}^{\mathrm{ctn}}_{\alpha\beta}(T)\rangle\!\rangle & =\frac{4}{(N-1)(N+2)}\left[A_{\alpha\beta}\int^{T}_{0}ds\int^{T}_{0}d\tau\tilde{K}(s-\tau)-B_{\alpha\beta}T^{2}\right]\nonumber \\
 & =\frac{4}{(N-1)(N+2)}\left[N^{2}A_{\alpha\beta}\int^{T}_{0}ds\int^{T}_{0}d\tau K(s-\tau)-B_{\alpha\beta}T^{2}\right].\label{eq:Q-ctn-phys-GOE}
\end{align}
where $A_{\alpha\beta}$ and $B_{\alpha\beta}$ are still defined
as in Eqs.~(\ref{eq:A-def}, \ref{eq:B-def}).

\section{Proof the self-average and typicality for integral SFF and SFF at
late and early times}

\subsection{Late-time fluctuations for the integral SFF and SFF}

Now we analyze the fluctuations of the integral of SFF and SFF at
large timescales. Our procedure is as follow: (i) Eigenvalue representation:
We rewrite $I_{H}(T)$, $I^{2}_{H}(T)$, $K_{H}(T)$ and $K^{2}_{H}(T)$
in terms of the sum of each contribution of the eigenvalues of $H$
(ii) Integral representation: We decompose the summation properly
so that upon taking the ensemble averages, each sum can be represented
by an integral involving the $n-$point spectral correlation function
defined as 
\begin{equation}
\rho^{(n)}_{N}(\lambda_{1},\,\lambda_{2},\,\cdots\lambda_{n})\equiv\int d\lambda_{n+1}\cdots d\lambda_{N}P_{N}(\lambda_{1},\,\lambda_{2},\,\cdots,\lambda_{N}),
\end{equation}
where $\int d\lambda\equiv\int^{\lambda_{\max}}_{\lambda_{\min}}d\lambda$
with $\lambda_{\min}\ge-\infty$ and $\lambda_{\max}\leq\infty$.
(iii) Late-time limit: While considering $N\gg1$ but fixed, take
the long-time limit and identify the non-vanishing terms. 

More specifically, let us now rewrite $I_{H}(T)$ defined in Eq.~(\ref{eq:K-I-def})
as
\begin{align}
I_{H}(T) & =\frac{1}{N^{2}}\sum_{k,\ell}\int^{1}_{0}(1-u)\mathrm{e}^{\mathrm{i}(\lambda_{k}-\lambda_{\ell})Tu}du\nonumber \\
 & =\frac{1}{2N}+\frac{1}{N^{2}T^{2}}\sum_{k,\ell,\,k\neq\ell}\frac{1+\mathrm{i}(\lambda_{k}-\lambda_{\ell})T-\mathrm{e}^{\mathrm{i}(\lambda_{k}-\lambda_{\ell})T}}{(\lambda_{k}-\lambda_{\ell})^{2}}\nonumber \\
 & =\frac{1}{2N}+\frac{\mathcal{C}^{(2)}_{H}(T)}{T^{2}},\label{eq:I-H-T}
\end{align}
where 
\begin{equation}
\mathcal{C}^{(2)}_{H}(T)\equiv\frac{1}{N^{2}}\sum_{k,\ell,\,k\neq\ell}\frac{1-\mathrm{e}^{\mathrm{i}(\lambda_{k}-\lambda_{\ell})T}}{(\lambda_{k}-\lambda_{\ell})^{2}}=\frac{2}{N^{2}}\sum_{k,\ell,\,k\neq\ell}\frac{\sin^{2}[(\lambda_{k}-\lambda_{\ell})T/2]}{(\lambda_{k}-\lambda_{\ell})^{2}}.
\end{equation}
Furthermore, it can be readily calculated that
\begin{align}
I^{2}_{H}(T) & =\left(\frac{1}{2N}+\frac{\mathcal{C}^{(2)}_{H}(T)}{T^{2}}\right)^{2}=\frac{1}{4N^{2}}+\frac{\mathcal{C}^{(2)}_{H}(T)}{NT^{2}}+\frac{[\mathcal{C}^{(2)}_{H}(T)]^{2}}{T^{4}}\nonumber \\
 & =\frac{1}{4N^{2}}+\frac{2}{N^{3}T^{2}}\sum_{k,\ell,\,k\neq\ell}\frac{\sin^{2}[(\lambda_{k}-\lambda_{\ell})T/2]}{(\lambda_{k}-\lambda_{\ell})^{2}}+\frac{4}{N^{4}T^{4}}\sum_{k,\ell,\,k\neq\ell}\sum_{m,n,\,m\neq n}\frac{\sin^{2}[(\lambda_{k}-\lambda_{\ell})T/2]\sin^{2}[(\lambda_{m}-\lambda_{n})T/2]}{(\lambda_{k}-\lambda_{\ell})^{2}(\lambda_{m}-\lambda_{n})^{2}}.
\end{align}
Using the fact that 
\begin{equation}
\sum_{\substack{k,\ell\\
k\neq\ell
}
}\sum_{\substack{m,n\\
m\neq n
}
}F_{k\ell mn}=\sum_{\substack{k,\ell,m,n\\
\text{distinct}
}
}F_{k\ell mn}+\sum_{\substack{k,\ell,n\\
\text{distinct}
}
}F_{k\ell kn}+\sum_{\substack{k,\ell,m\\
\text{distinct}
}
}F_{k\ell mk}+\sum_{\substack{k,\ell,n\\
\text{distinct}
}
}F_{k\ell\ell n}+\sum_{\substack{k,\ell,m\\
\text{distinct}
}
}F_{k\ell m\ell}+\sum_{k\neq\ell}F_{k\ell k\ell}+\sum_{k\neq\ell}F_{k\ell\ell k},\label{eq:sum-decomp}
\end{equation}
one can decompose the quadruple sum as
\begin{align}
 & \sum_{k,\ell,\,k\neq\ell}\sum_{m,n,\,m\neq n}\frac{\sin^{2}[(\lambda_{k}-\lambda_{\ell})T/2]\sin^{2}[(\lambda_{m}-\lambda_{n})T/2]}{(\lambda_{k}-\lambda_{\ell})^{2}(\lambda_{m}-\lambda_{n})^{2}}\nonumber \\
= & \sum_{\substack{k,\ell,m,n\\
\text{distinct}
}
}\frac{\sin^{2}[(\lambda_{k}-\lambda_{\ell})T/2]\sin^{2}[(\lambda_{m}-\lambda_{n})T/2]}{(\lambda_{k}-\lambda_{\ell})^{2}(\lambda_{m}-\lambda_{n})^{2}}+4\sum_{\substack{k,\ell,n\\
\text{distinct}
}
}\frac{\sin^{2}[(\lambda_{k}-\lambda_{\ell})T/2]\sin^{2}[(\lambda_{k}-\lambda_{n})T/2]}{(\lambda_{k}-\lambda_{\ell})^{2}(\lambda_{k}-\lambda_{n})^{2}}+2\sum_{k\neq\ell}\frac{\sin^{4}[(\lambda_{k}-\lambda_{\ell})T/2]}{(\lambda_{k}-\lambda_{\ell})^{4}}.
\end{align}
After taking the ensemble average rule 
\begin{align}
\sum_{\substack{k,\ell,m,n\\
\text{all distinct}
}
}\to & N(N-1)(N-2)(N-3)\int d\lambda_{1}d\lambda_{2}d\lambda_{3}d\lambda_{4}\rho^{(4)}_{N}(\lambda_{1},\lambda_{2},\lambda_{3},\,\lambda_{4}),\\
\sum_{\substack{k,\ell,n\\
\text{all distinct}
}
}\to & N(N-1)(N-2)\int d\lambda_{1}d\lambda_{2}d\lambda_{3}\rho^{(3)}_{N}(\lambda_{1},\lambda_{2},\lambda_{3}),\\
\sum_{k\neq\ell}\to & N(N-1)\int d\lambda_{1}d\lambda_{2}\rho^{(2)}_{N}(\lambda_{1},\lambda_{2}),
\end{align}
we find
\begin{equation}
I(T)\equiv\langle\!\langle I_{H}(T)\rangle\!\rangle=\frac{1}{2N}+\frac{2}{T^{2}}\left(1-\frac{1}{N}\right)\int d\lambda_{1}d\lambda_{2}\rho^{(2)}_{N}(\lambda_{1},\lambda_{2})\frac{\sin^{2}[(\lambda_{1}-\lambda_{2})T/2]}{(\lambda_{1}-\lambda_{2})^{2}},\label{eq:IH-ave}
\end{equation}
\begin{align}
\langle\!\langle I^{2}_{H}(T)\rangle\!\rangle & =\frac{1}{4N^{2}}+\frac{N-1}{N^{2}T^{2}}\int d\lambda_{1}d\lambda_{2}\rho^{(2)}_{N}(\lambda_{1},\lambda_{2})\frac{\sin^{2}[(\lambda_{1}-\lambda_{2})T/2]}{(\lambda_{1}-\lambda_{2})^{2}}\nonumber \\
 & +\frac{2(N-1)}{N^{3}T^{4}}\int d\lambda_{1}d\lambda_{2}\rho^{(2)}_{N}(\lambda_{1},\lambda_{2})\frac{\sin^{4}[(\lambda_{1}-\lambda_{2})T/2]}{(\lambda_{1}-\lambda_{2})^{4}}\nonumber \\
 & +\frac{4(N-1)(N-2)}{N^{3}T^{4}}\int d\lambda_{1}d\lambda_{2}d\lambda_{3}\rho^{(3)}_{N}(\lambda_{1},\lambda_{2},\lambda_{3})\frac{\sin^{2}[(\lambda_{1}-\lambda_{2})T/2]\sin^{2}[(\lambda_{1}-\lambda_{3})T/2]}{(\lambda_{1}-\lambda_{2})^{2}(\lambda_{1}-\lambda_{3})^{2}}\nonumber \\
 & +\frac{4(N-1)(N-2)(N-3)}{N^{3}T^{4}}\int d\lambda_{1}d\lambda_{2}d\lambda_{3}d\lambda_{4}\rho^{(4)}_{N}(\lambda_{1},\lambda_{2},\lambda_{3},\,\lambda_{4})\frac{\sin^{2}[(\lambda_{1}-\lambda_{2})T/2]\sin^{2}[(\lambda_{3}-\lambda_{4})T/2]}{(\lambda_{1}-\lambda_{2})^{2}(\lambda_{3}-\lambda_{4})^{2}}.\label{eq:I2H-ave}
\end{align}
Following similar procedure, it can be readily calculate that 
\begin{equation}
K_{H}(T)=\frac{1}{N^{2}}\sum_{k,\ell}\mathrm{e}^{\mathrm{i}(\lambda_{k}-\lambda_{\ell})T}=\frac{1}{N}+\frac{1}{N^{2}}\sum_{k,\ell,\,k\neq\ell}\cos[(\lambda_{k}-\lambda_{\ell})T],\label{eq:K-H-T}
\end{equation}
and 
\begin{equation}
K^{2}_{H}(T)=\frac{1}{N^{2}}+\frac{2}{N^{3}}\sum_{k,\ell,\,k\neq\ell}\cos[(\lambda_{k}-\lambda_{\ell})T]+\frac{1}{N^{4}}\sum_{k,\ell,\,k\neq\ell}\sum_{m,n,\,m\neq n}\cos[(\lambda_{k}-\lambda_{\ell})T]\cos[(\lambda_{m}-\lambda_{n})T].
\end{equation}
According to Eq.\ (\ref{eq:sum-decomp}), we find
\begin{multline}
\sum_{k,\ell,\,k\neq\ell}\sum_{m,n,\,m\neq n}\cos[(\lambda_{k}-\lambda_{\ell})T]\cos[(\lambda_{m}-\lambda_{n})T]=\sum_{\substack{k,\ell,m,n\\
\mathrm{distinct}
}
}\cos[(\lambda_{k}-\lambda_{\ell})T]\cos[(\lambda_{m}-\lambda_{n})T]\\
+4\sum_{\substack{k,l,m\\
\mathrm{distinct}
}
}\cos[(\lambda_{k}-\lambda_{\ell})T]\cos[(\lambda_{m}-\lambda_{k})T]+2\sum_{k\neq l}\cos^{2}[(\lambda_{k}-\lambda_{\ell})T].
\end{multline}
Therefore 
\begin{equation}
K(T)\equiv\langle\!\langle K_{H}(T)\rangle\!\rangle=\frac{1}{N}+\left(1-\frac{1}{N}\right)\int d\lambda_{1}d\lambda_{2}\rho^{(2)}_{N}(\lambda_{1},\lambda_{2})\cos[(\lambda_{1}-\lambda_{2})T],\label{eq:KH-ave}
\end{equation}
and 
\begin{align}
\langle\!\langle K^{2}_{H}(T)\rangle\!\rangle & =\frac{1}{N^{2}}+\frac{2(N-1)}{N^{2}}\int d\lambda_{1}d\lambda_{2}\rho^{(2)}_{N}(\lambda_{1},\lambda_{2})\cos[(\lambda_{1}-\lambda_{2})T]+\frac{2(N-1)}{N^{3}}\int d\lambda_{1}d\lambda_{2}\rho^{(2)}_{N}(\lambda_{1},\lambda_{2})\cos^{2}[(\lambda_{1}-\lambda_{2})T]\nonumber \\
 & +\frac{4(N-1)(N-2)}{N^{3}}\int d\lambda_{1}d\lambda_{2}d\lambda_{3}\rho^{(3)}_{N}(\lambda_{1},\lambda_{2},\lambda_{3})\cos[(\lambda_{1}-\lambda_{2})T]\cos[(\lambda_{1}-\lambda_{3})T]\nonumber \\
 & +\frac{(N-1)(N-2)(N-3)}{N^{3}}\int d\lambda_{1}d\lambda_{2}d\lambda_{3}d\lambda_{4}\rho^{(4)}_{N}(\lambda_{1},\lambda_{2},\lambda_{3},\,\lambda_{4})\cos[(\lambda_{1}-\lambda_{2})T]\cos[(\lambda_{3}-\lambda_{4})T].\label{eq:K2H-ave}
\end{align}
Now let us analyze the asymptotic behavior of Eqs.\ (\ref{eq:IH-ave},\ \ref{eq:I2H-ave},\ \ref{eq:KH-ave},\ \ref{eq:K2H-ave})
when $T$ becomes large. Let us introduce 
\begin{equation}
\begin{aligned}\sigma^{(2)}_{N}(\lambda) & =\int d\mu\,\rho^{(2)}_{N}\!\left(\mu+\frac{\lambda}{2},\mu-\frac{\lambda}{2}\right),\\
\sigma^{(3)}_{N}(\lambda_{1},\lambda_{2}) & =\int d\mu\,\rho^{(3)}_{N}(\mu,\mu-\lambda_{1},\mu-\lambda_{2}),\\
\sigma^{(4)}_{N}(\lambda_{1},\lambda_{2}) & =\int d\mu_{1}\,d\mu_{2}\,\rho_{4}\!\left(\mu_{1}+\frac{\lambda_{1}}{2},\mu_{1}-\frac{\lambda_{1}}{2},\mu_{2}+\frac{\lambda_{2}}{2},\mu_{2}-\frac{\lambda_{2}}{2}\right).
\end{aligned}
\end{equation}
Then we find 
\begin{equation}
\int d\lambda_{1}d\lambda_{2}\rho^{(2)}_{N}(\lambda_{1},\lambda_{2})\frac{\sin^{2}[(\lambda_{1}-\lambda_{2})T/2]}{(\lambda_{1}-\lambda_{2})^{2}}\to\frac{\pi T}{2}\int d\lambda_{1}d\lambda_{2}\rho^{(2)}_{N}(\lambda_{1},\lambda_{2})\delta(\lambda_{1}-\lambda_{2})=\frac{\pi}{2}T\sigma^{(2)}_{N}(0)+O(1),
\end{equation}
\begin{equation}
\int d\lambda_{1}d\lambda_{2}\rho^{(2)}_{N}(\lambda_{1},\lambda_{2})\frac{\sin^{4}[(\lambda_{1}-\lambda_{2})T/2]}{(\lambda_{1}-\lambda_{2})^{4}}\to\frac{\pi}{12}T^{3}\sigma^{(2)}_{N}(0)+O(1),
\end{equation}
where we have used the fact that 
\begin{equation}
\frac{\sin^{2}(T\lambda)}{\lambda^{2}}\to\frac{\pi}{2}T\delta(0),\quad\frac{\sin^{4}(\lambda T)}{\lambda^{4}}\to\frac{\pi}{12}T^{3}\delta(0)
\end{equation}
as $T\to\infty$ since $\int^{\infty}_{-\infty}\frac{\sin^{2}(T\lambda/2)}{\lambda^{2}}=\frac{\pi}{2}T$
and $\int^{\infty}_{-\infty}\frac{\sin^{4}(T\lambda/2)}{\lambda^{4}}=\frac{\pi}{12}T^{3}$.
Similarly, we find where 
\begin{align}
\int d\lambda_{1}d\lambda_{2}d\lambda_{3}\rho^{(3)}_{N}(\lambda_{1},\lambda_{2},\lambda_{3})\frac{\sin^{2}[(\lambda_{1}-\lambda_{2})T/2]\sin^{2}[(\lambda_{1}-\lambda_{3})T/2]}{(\lambda_{1}-\lambda_{2})^{2}(\lambda_{1}-\lambda_{3})^{2}} & \to\frac{\pi^{2}}{4}T^{2}\sigma^{(3)}_{N}(0,0)+O(1),\\
\int d\lambda_{1}d\lambda_{2}d\lambda_{3}d\lambda_{4}\rho^{(4)}_{N}(\lambda_{1},\lambda_{2},\lambda_{3},\,\lambda_{4})\frac{\sin^{2}[(\lambda_{1}-\lambda_{2})T/2]\sin^{2}[(\lambda_{3}-\lambda_{4})T/2]}{(\lambda_{1}-\lambda_{2})^{2}(\lambda_{3}-\lambda_{4})^{2}} & \to\frac{\pi^{2}}{4}T^{2}\sigma^{(4)}_{N}(0,0)+O(1).
\end{align}
On the other hand, according to the Riemman-Lebesgue theorem, 
\begin{equation}
\int f(\lambda)\mathrm{e}^{\mathrm{i}\lambda T}=\sum^{m}_{k=1}\frac{(-1)^{k}}{(\mathrm{i}T)^{k+1}}f^{(k)}(\lambda)|^{\lambda=\lambda_{\max}}_{\lambda=\lambda_{\min}}+O\left(\frac{1}{T^{m+1}}\right),
\end{equation}
 one can easily find 
\begin{align}
\int d\lambda_{1}d\lambda_{2}\rho^{(2)}_{N}(\lambda_{1},\lambda_{2})\cos[(\lambda_{1}-\lambda_{2})T] & \sim O\left(\frac{1}{T}\right),\\
\int d\lambda_{1}d\lambda_{2}\rho^{(2)}_{N}(\lambda_{1},\lambda_{2})\cos^{2}[(\lambda_{1}-\lambda_{2})T] & \sim\frac{1}{2}+O\left(\frac{1}{T}\right),\\
\int d\lambda_{1}d\lambda_{2}d\lambda_{3}\rho^{(3)}_{N}(\lambda_{1},\lambda_{2},\lambda_{3})\cos[(\lambda_{1}-\lambda_{2})T]\cos[(\lambda_{1}-\lambda_{3})T] & \sim O\left(\frac{1}{T}\right),\\
\int d\lambda_{1}d\lambda_{2}d\lambda_{3}d\lambda_{4}\rho^{(4)}_{N}(\lambda_{1},\lambda_{2},\lambda_{3},\,\lambda_{4})\cos[(\lambda_{1}-\lambda_{2})T]\cos[(\lambda_{3}-\lambda_{4})T] & \sim O\left(\frac{1}{T}\right).
\end{align}
Therefore in the limit of $T\gg N\gg1$, we obtain
\begin{align}
\langle\!\langle I_{H}(T)\rangle\!\rangle & \to\frac{1}{2N}+\frac{\pi}{2T}\sigma^{(2)}_{N}(0)+O(1),\\
\langle\!\langle K_{H}(T)\rangle\!\rangle & \to\frac{1}{N}+O\left(\frac{1}{T}\right),\\
\langle\!\langle I^{2}_{H}(T)\rangle\!\rangle & \to\frac{1}{4N^{2}}+\frac{\pi}{2NT}\sigma^{(2)}_{N}(0)+\frac{\pi}{6N^{2}T}\sigma^{(2)}_{N}(0)+\frac{\pi^{2}}{NT^{2}}\sigma^{(3)}_{N}(0,0)+\frac{\pi^{2}}{4T^{2}}\sigma^{(4)}_{N}(0,0),\\
\langle\!\langle K^{2}_{H}(T)\rangle\!\rangle & \to\frac{2}{N^{2}}+O\left(\frac{1}{T}\right).
\end{align}
It is then clear that at late times $T\gg N$, for the SFF, 
\begin{equation}
\mathrm{RV}_{N}[K_{H}(T)]\equiv\frac{\langle\!\langle K^{2}_{H}(T)\rangle\!\rangle}{\langle\!\langle K_{H}(T)\rangle\!\rangle^{2}}-1\to1,
\end{equation}
which implies that the SFF is not self-averaging, while for the integral
SFF, we find 
\begin{equation}
\mathrm{RV}_{N}[I_{H}(T)]=\frac{\langle\!\langle I^{2}_{H}(T)\rangle\!\rangle}{\langle\!\langle I_{H}(T)\rangle\!\rangle^{2}}-1\to0.
\end{equation}

\subsection{Concentration inequality}

Let us first introduce briefly the background of concentration inequality
in order to make the Supplemental Material self-contained. The details
can be found in Ref.\ \citep{anderson2010anintroduction}. A probability
measure $\mu$ on a Riemannian manifold $M$ of dimension $m$ (e.g.,
$M=\mathbb{R}^{m}$) is said to satisfy the logarithmic Sobolev inequality
(LSI), if any non-negative function on the manifold $g:M\to\mathbb{R}$,
it satisfies 
\begin{equation}
\int g(x)\log\frac{g(x)}{\mathbb{E}_{\mu}g}d\mu(x)\leq\frac{c_{\mathrm{LSI}}}{2}\int\frac{\|\nabla g(x)\|^{2}}{g(x)}d\mu(x).\label{eq:LSI-def1}
\end{equation}
Without loss of generality, one can further normalize $g$ such that
$\int g(x)d\mu(x)=1$ or then Eq.\ (\ref{eq:LSI-def1}) become 
\begin{align}
\int g(x)\log g(x)d\mu(x) & \leq\frac{c_{\mathrm{LSI}}}{2}\int\frac{\|\nabla g(x)\|^{2}}{g(x)}d\mu(x).\label{eq:LSI-normalized}
\end{align}
Then the l.h.s. is just the Shannon entropy of the distribution $g(x)$
while the r.h.s is classical Fisher information contained in the probability
distribution $g(x)$. Similarly, the l.h.s. of Eq.\ (\ref{eq:LSI-def1})
can be also interpreted as the relative entropy between the unnormalized
distribution $g(x)$ and the uniform distribution $\mathbb{E}_{\mu}g$
while the r.h.s. can be interpreted as the classical Fisher information
of the unnormalized probability distribution $g(x)$

Upon setting $g=f^{2}$, one can obtain an alternative form of the
LSI,
\begin{equation}
\int f^{2}(x)\log\frac{f^{2}(x)}{\mathbb{E}_{\mu}f^{2}}d\mu(x)\leq c_{\mathrm{LSI}}\int\|\nabla f(x)\|^{2}d\mu(x).\label{eq:LSI-def2}
\end{equation}
In practice, one can employ the following lemmas to identify a given
probability to satisfy LSI and the corresponding coefficient $c_{\mathrm{LSI}}$. 
\begin{lem}
\textup{\label{lem:Bochner-Bakry-Emergy}Consider a function $\Phi:M\to\mathbb{R}$,
the probability measure 
\begin{equation}
d\mu_{\Phi}\equiv\frac{1}{Z}\mathrm{e}^{-\Phi(x)}d\mathrm{vol}(x),\label{eq:dmu-Phi}
\end{equation}
where $d\mathrm{vol}(x)=\sqrt{\det g(x)}dx_{1}\wedge dx_{2}\wedge\cdots\wedge dx_{m}$,
satisfies LSI with constant $c_{\mathrm{LSI}}$ if the following inequality
holds
\begin{equation}
[\mathrm{Ric}+\mathrm{Hess}(\Phi)](\mathcal{X},\mathcal{X})\ge\frac{1}{c_{\mathrm{LSI}}}g_{x}(\mathcal{X},\mathcal{X})\label{eq:BBE}
\end{equation}
for any tangent vector $\mathcal{X}\in T_{x}M$at any point $x\in M$,
where $\mathrm{Ric}(\Phi)$ is the Ricci tensor for for the Riemannian
manifold $M$ and $\mathrm{Hess}(\Phi)$ is the Hessian of the function
$\Phi$.}
\end{lem}
When $M=\mathbb{R}^{m}$, the Ricci tensor is zero since flat space
has no curvature and the metric is the Euclidean metric. Eq.\ \ref{eq:dmu-Phi}
and Eq.\ \ref{eq:BBE} becomes 
\begin{equation}
d\mu_{\Phi}\equiv\frac{1}{Z}\mathrm{e}^{-\Phi(x)}dx_{1}dx_{2}\cdots dx_{m},\qquad\mathrm{Hess}(\Phi)\ge\frac{1}{c_{\mathrm{LSI}}}\mathbb{I},
\end{equation}
where the matrix element of the Hessian of $\Phi(x)$ is $[\mathrm{Hess}(\Phi)]_{ij}=\partial_{i}\partial_{j}\Phi(x)$.

Note that the condition Eq.\ (\ref{eq:BBE}) is a sufficient condition
for the LSI to hold. It is by no means necessary. 
\begin{lem}
\textup{\label{lem:LSI-Product}Given $n$-probability measure $\{\mu^{(i)}\}^{n}_{i=1}$
on $n$-Riemannian manifold $\{M_{i}\}^{n}_{i=1}$ with each of them
satisfying the LSI with constant $c^{(i)}_{\mathrm{LSI}}$, the product
measure $\mu^{(M)}=\mu^{(1)}\otimes\mu^{(2)}\otimes\cdots\otimes\mu^{(n)}$
on the manifold $M=M_{1}\times M_{2}\times\cdots\times M_{n}$ also
satisfy the LSI with the constant $c^{(M)}_{\mathrm{LSI}}=\max\{c^{(i)}_{\mathrm{LSI}}\}$.}
\end{lem}
We consider the GUE, where 
\begin{equation}
d\mu(H)\propto\mathrm{e}^{-\frac{N\mathrm{Tr}H^{2}}{2}}\prod_{i}dH_{ii}\prod_{j\leq k}d\mathrm{Re}H_{jk}d\mathrm{Im}H_{jk}.
\end{equation}
The manifold is essentially $\mathbb{R}^{N^{2}}$, which does not
have any curvature. One can readily find that 
\begin{equation}
\mathrm{Hess}(\mathrm{e}^{-\frac{N\mathrm{Tr}H^{2}}{2}})=N\mathrm{diag}(\underbrace{1,\cdots1}_{\#=N},\underbrace{2,2\cdots2}_{\#=N(N-1)})\ge N\mathbb{I},
\end{equation}
and therefore
\begin{equation}
c^{(\mathrm{GUE})}_{\mathrm{LSI}}=\frac{1}{N}.
\end{equation}
Realistic disordered Hamiltonian ensemble is typically implemented
through 
\begin{equation}
H[x]=H_{0}+\sum_{A}h_{A}\sigma_{A},
\end{equation}
where without loss of generality we assume $\sigma_{A}$ is Hermitian,
traceless and its support $A$ only contains a finite number of sites
and $h_{A}\in[-h,h]$. For one-dimensional uniform distribution on
an interval $[a,\,b]$, it can be shown that\ \citep{ghang2013thesharp}
\begin{equation}
\frac{2\pi^{2}}{\ell^{2}}\int g(x)\log\frac{g(x)}{\mathbb{E}_{\mu}g(x)}d\mu(x)\leq\int|\frac{g^{\prime}(x)|^{2}}{g(x)}d\mu(x),\label{eq:LSI-uniform}
\end{equation}
where $\ell=b-a$. The idea of proving Eq.\ (\ref{eq:LSI-uniform})
is to exploit the LSI for the uniform Haar measure on a unit circle\ \citep{weissler1980logarithmic,emery1987asimple}
\begin{equation}
2\pi^{2}\int_{S^{1}}f^{2}(x)\frac{\log f^{2}(x)}{\mathbb{E}_{\mu}f^{2}(x)}d\mu_{\mathrm{Haar}}(x)\leq\int_{S^{1}}|f^{\prime}(x)|^{2}d\mu_{\mathrm{Haar}}(x),
\end{equation}
and map the integral on the unit circuit to the interval $[0,1]$.
It is then clear from that Eq.\ (\ref{eq:LSI-uniform}) and Lemma\ \ref{lem:LSI-Product}
that the constant for the uniform Hamiltonian ensemble is 
\begin{equation}
c^{(\mathrm{uni})}_{\mathrm{LSI}}=\frac{4h^{2}}{\pi^{2}}.
\end{equation}

A smooth, real-valued function $f:M\to\mathbb{R}$ is the called Lipschitz
function if $\exists$ a constant $L_{f}$ such that 
\begin{equation}
|f(x)-f(y)|\leq L_{f}d(x,y)\leq\mathrm{sup}_{x\in M}\|\nabla f\|d(x,y),
\end{equation}
where $d(x,y)$ is the length of the geodesic connecting $x$ and
$y$. Clearly, if 
\begin{equation}
\mathrm{sup}_{x\in M}\|\nabla f\|<\infty,
\end{equation}
then we know 
\begin{equation}
L_{f}\leq\mathrm{sup}_{x\in M}\|\nabla f\|,
\end{equation}
and $f$ must be Lipschitz. When $M$ is compact, $\mathrm{sup}_{x\in M}\|\nabla f\|<\infty$
must hold. Therefore, a smooth function on a compact manifold must
be Lipschitz. The concentration inequality upper bounds the fluctuations
in the following sense: if given a probability measure and a real-valued
function on a Riemannian manifold satisfies the LSI and Lipschitz
condition respectively, then we have the concentration inequality
\begin{equation}
\mu\left(|f(x)-\mathbb{E}_{\mu}f|\ge\varepsilon\right)\leq2\mathrm{e}^{-\frac{\varepsilon^{2}}{2c_{\mathrm{LSI}}L^{2}_{f}}}.
\end{equation}
Upon replacing $\varepsilon$ with $\varepsilon\mathbb{E}_{\mu}f$,
one obtain the relative concentration inequality:
\begin{equation}
\mu\left(\bigg|\frac{f(x)}{\mathbb{E}_{\mu}f}-1\bigg|\ge\varepsilon\right)\leq2\mathrm{e}^{-\frac{\varepsilon^{2}(\mathbb{E}_{\mu}f)^{2}}{2c_{\mathrm{LSI}}L^{2}_{f}}}.
\end{equation}

\subsection{Concentration inequality for the integral SFF and SFF}

For the sake of simplicity, we shall suppress the dependence on $T$
in $K_{H}(T)$ and $I_{H}(T)$ in the following discussion. We only
recover the $T$-dependence at the end of the discussion. First, let
us first compute the Lipschitz constants. 
\begin{align}
|\langle\nabla_{H}K_{H},\,X\rangle| & \equiv\bigg|\frac{dK_{H+\epsilon X}}{d\epsilon}\big|_{\epsilon=0}\bigg|=\frac{1}{N^{2}}\bigg|\frac{d|\mathrm{Tr}\mathrm{e}^{\mathrm{i}(H+\epsilon X)T}|^{2}}{d\epsilon}\bigg|=\frac{2T}{N^{2}}|\mathrm{Im}[\mathrm{Tr}(\mathrm{e}^{\mathrm{i}HT})\mathrm{Tr}(X\mathrm{e}^{-\mathrm{i}HT})]|\nonumber \\
 & \leq\frac{2T}{N^{2}}|\mathrm{Tr}(\mathrm{e}^{\mathrm{i}HT})\mathrm{Tr}(X\mathrm{e}^{-\mathrm{i}HT})|\leq\frac{2T}{\sqrt{N}}\|X\|_{\mathrm{HS}},\:\forall X,
\end{align}
where we have used the Cauchy-Schwarz inequality $|\mathrm{Tr}(X^{\dagger}Y)|\leq\|X\|_{\mathrm{HS}}\|Y\|_{\mathrm{HS}}$. 

On the other hand, we note that $|\langle\nabla_{H}K_{H},\,X\rangle|\leq\|\nabla_{H}K_{H}\|_{\mathrm{HS}}\|X\|_{\mathrm{HS}},\:\forall X$.
Therefore, when the Hamiltonian is taken from the GUE, we obtain
\begin{equation}
L^{(\mathrm{GUE})}_{K}\leq\|\nabla_{H}K_{H}\|_{\mathrm{HS}}\leq\frac{2T}{\sqrt{N}}.
\end{equation}
When the Hamiltonian is taken from the uniform ensemble, according
to the chain rule, we find
\begin{equation}
\partial_{A}K_{\bm{h}}=\langle\nabla_{H}K_{H},\,\sigma_{A}\rangle,
\end{equation}
where $\bm{h}\equiv\{h_{A}\}_{A}$ and $\partial_{A}\equiv\partial_{h_{A}}.$
Thus 
\begin{equation}
|\sum_{A}y_{A}\partial_{A}K_{\bm{h}}|=|\langle\nabla_{H}K_{H},\,\sum_{A}y{}_{A}\sigma_{A}\rangle|\leq\|\nabla_{H}K_{H}\|_{\mathrm{HS}}\|\sum_{A}y_{A}\sigma_{A}\|_{\mathrm{HS}}.
\end{equation}
By definition 
\begin{equation}
\|\sum_{A}y_{A}\sigma_{A}\|^{2}_{\mathrm{HS}}=\sum_{A,B}y_{A}\mathrm{Tr}(\sigma_{A}\sigma_{B})y_{B}\leq\lambda_{\max}(\Xi)\|\bm{y}\|^{2},
\end{equation}
where the matrix element of the matrix $\Xi$
\begin{equation}
\Xi_{AB}=\mathrm{Tr}(\sigma_{A}\sigma_{B}).
\end{equation}
Note that if $\{\sigma_{A}\}$ are single-site Pauli string operators,
the it clear that $\lambda_{\max}(\Xi)=2$. For general local disorder
coupling, where $A$ only supports on a finite number of lattices,
we know $\lambda_{\max}(\Xi)\sim O(1)$. 

Upon noting $|\sum_{A}y_{A}\partial_{A}K_{\bm{h}}|\leq\|\bm{y}\|\|\nabla K_{\bm{h}}\|$
for all $\bm{y}$, we find 
\begin{equation}
\|\nabla K_{\bm{h}}\|\leq\|\nabla_{H}K_{H}\|_{\mathrm{HS}}\sqrt{\lambda_{\max}(\Xi)}\leq\frac{2T\sqrt{\lambda_{\max}(\Xi)}}{\sqrt{N}}.
\end{equation}
Then for the integral of SFF, it is clear that 
\begin{equation}
\|\nabla_{H}I_{H}\|=\|\int^{1}_{0}(1-u)\nabla K_{H}(uT)du\|\leq\int^{1}_{0}(1-u)\|\nabla K_{H}(uT)\|du=\frac{T}{\sqrt{N}},
\end{equation}
\begin{equation}
\|\nabla_{\bm{h}}I_{H}\|\leq\int^{1}_{0}(1-u)\|\nabla K_{\bm{h}}\|du=\frac{T\sqrt{\lambda_{\max}(\Xi)}}{\sqrt{N}}.
\end{equation}
Therefore, when the Hamiltonian is taken from GUE, the concentration
inequalities for the SFF and the integral SFF become 
\begin{align}
\mu\left(\bigg|\frac{K_{H}(T)}{K(T)}-1\bigg|\ge\varepsilon\right) & \leq2\mathrm{e}^{-\frac{\varepsilon^{2}N^{2}K^{2}(T)}{8T^{2}}},\\
\mu\left(\bigg|\frac{I_{H}(T)}{I(T)}-1\bigg|\ge\varepsilon\right) & \leq2\mathrm{e}^{-\frac{\varepsilon^{2}N^{2}I^{2}(T)}{2T^{2}}},
\end{align}
and for the uniform Hamiltonian ensemble, the concentration inequalities
becomes 
\begin{align}
\mu\left(\bigg|\frac{K_{\bm{h}}(T)}{K(T)}-1\bigg|\ge\varepsilon\right) & \leq2\mathrm{e}^{-\frac{\pi^{2}\varepsilon^{2}NK^{2}(T)}{32h^{2}T^{2}\lambda_{\max}(\Sigma)}},\\
\mu\left(\bigg|\frac{I_{\bm{h}}(T)}{I(T)}-1\bigg|\ge\varepsilon\right) & \leq2\mathrm{e}^{-\frac{\pi^{2}\varepsilon^{2}NI^{2}(T)}{8h^{2}T^{2}\lambda_{\max}(\Sigma)}},
\end{align}
where we recovered the time-dependence in the SFF and integral SFF. 

Apparently, the concentration inequality implies that the SFF and
integral SFF are typical at early time scales. At the late-time scale
with $T\gg N$, $K(T),\,I(T)\sim1/N$, therefore they becomes trivial.
More generally, for GUE, the SFF and its integral are typical respectively
at time $T$ if $\exists\delta>0$, such that 
\begin{equation}
\frac{K(T)}{T}\ge\frac{1}{N^{1-\delta}},\quad\frac{I(T)}{T}\ge\frac{1}{N^{1-\delta}},
\end{equation}
while for the uniformly disordered Hamiltonian ensemble, the SFF and
its integral are typical respectively at time $T$ if $\exists\delta>0$,
such that 
\begin{equation}
\frac{K(T)}{T}\ge\frac{1}{N^{1/2-\delta}},\quad\frac{I(T)}{T}\ge\frac{1}{N^{1/2-\delta}}.
\end{equation}

\section{\label{sec:Numerical-Calculation}Details on the numerical calculations}

In this section, we provide a detailed description of the numerical
calculations for Fig.\ \ref{fig:Numerics}, the relative variance
of the QFI, and the saturation values. System parameters are described
in the main text. 
\begin{itemize}
\item Mean level spacing distribution, $P(s)$, and mean adjacent gap ratio,
${\displaystyle \langle\tilde{r}}\rangle$ shown in Fig.\ \ref{fig:Numerics}(a).
\\
To study the mean level spacing distribution, we first need to perform
an unfolding procedure to render the local density of states to be
uniform. This allows us to compare spectral statistics of physical
systems with the universal predictions of RMT. The unfolding procedure
is as follows\ \citep{haake2001quantum}: we discard the lowest and
highest $5\%$ of the energy levels from the ordered spectrum, which
contains $N$ energy levels, and perform our analysis in the bulk
of the spectrum. After discarding the edges, we divide the remaining
spectrum into consecutive intervals labeled by $\ell$, each containing
$M=10$ nearest-neighbor energy-level spacings. We then compute the
mean level spacing $\Delta^{(\ell)}$ within each interval, ${\displaystyle {\textstyle \Delta^{(\ell)}=\frac{1}{M-1}\sum^{M-1}_{i=1}\delta^{(\ell)}_{i+1,i}}}$
with ${\textstyle \delta^{(\ell)}_{i+1,i}=\lambda^{(\ell)}_{i+1}-\lambda^{(\ell)}_{i}}.$
Next we normalize the spacings within each interval by the corresponding
mean level spacing ${\displaystyle {\textstyle s_{i}=\frac{\delta^{(\ell)}_{i+1,i}}{\Delta^{(\ell)}}}}$,
and finally compute the mean level spacing distribution, $P(s)$.
On the other hand, the mean adjacent gap ratio provides an alternative
to diagnose whether the system is in the chaotic regime without requiring
the unfolding procedure\ \citep{atas2013distribution}. We simply
compute ${\displaystyle {\textstyle \tilde{r_{i}}=\mbox{Min}{\displaystyle \left(r_{i},r^{-1}_{i}\right)}}}$with
${\textstyle r_{i}=\frac{\lambda_{i+2}-\lambda_{i+1}}{\lambda_{i+1}-\lambda_{i}}}$
and then obtain the mean adjacent gap ratio as ${\textstyle {\displaystyle \langle\tilde{r}}\rangle={\textstyle {\textstyle {\textstyle {\displaystyle {\textstyle \frac{1}{N-2}}}}{\textstyle \sum}}^{N-2}_{i=1}\tilde{r_{i}}}}$.
Therefore, our physical systems are in the chaotic regime for the
set of chosen parameters. For disordered FRF-XXX model, an average
over $50$ realizations was considered while a single realization
is considered for the clean MFI model . The system size is $N=2^{12}$
for both models. 
\item Numerical calculation of the smoothed SFF $K(T)$ shown in Fig.\ \ref{fig:Numerics}(b-d).
\\
For the analysis of the smoothed SFF we first diagonalize the Hamiltonian
matrix to obtain its energy spectrum and then compute ${\textstyle {\textstyle K(T)={\displaystyle {\textstyle \frac{1}{N^{2}}}{\displaystyle {\displaystyle |}{\textstyle \sum_{j}\text{e}^{-\text{i}\lambda_{j}T}}{\displaystyle |}}^{2}}}}$.
In addition, because its time evolution exhibits large fluctuations,
an average is required to reveal the ramp, a characteristic feature
of chaotic systems. For Figs.\ \ref{fig:Numerics}(b) and\ \ref{fig:Numerics}(c)
an average over $500$ realizations of the Hamiltonians was performed
in order to detect the ramp. On the other hand, for Fig.\ \ref{fig:Numerics}(d),
the SFF was computed for a single realization, and due to its large
fluctuations, the ramp is difficult to detect. The system size is
$N=2^{11}$ in all three panels (b-d) in Fig.\ \ref{fig:Numerics}.
\item Analytic results for the average QFI $T^{-2}\langle\!\langle\bar{Q}\rangle\!\rangle$
shown in Figs.\ \ref{fig:Numerics}(b-d). \\
Based on our analytical expressions Eq.\ (\ref{eq:Q-ctn}), Eq.(\ref{eq:Q-ctn-phy-H}),
Eq.(\ref{eq:Q-ctn-GOE}), and Eq.(\ref{eq:Q-ctn-phys-GOE}), we first
perform the time integration of the smoothed SFF, then compute the
miscellaneous constants $C_{\alpha\beta}$, $A_{\alpha\beta}$, and
$B_{\alpha\beta}$ mentioned in the main text. and finally evaluate
Eq.\ (\ref{eq:Q-ctn}). The additional labels ``deterministic''
and ``randomized'' in Figs.\ \ref{fig:Numerics}(c) and\ \ref{fig:Numerics}(d)
refer to the choice of the operator $P$ appearing in the sensing
Hamiltonian. In the deterministic case $P=\sigma_{x}\otimes\mathbb{I}^{\otimes(L-1)}$
, whereas in the randomized case $P$ is replaced by $V^{\dagger}PV$,
where $V$ is a matrix of Haar-random vectors stored column-wise.
In addition, for the deterministic protocol we compute the coefficients
$A_{\alpha\beta}$, and $B_{\alpha\beta}$. by averaging over $500$
realizations of Hamiltonians.
\item Numerical calculation for the average QFI, $T^{-2}\langle\!\langle\bar{Q}\rangle\!\rangle$
shown in Figs.\ \ref{fig:Numerics}(b-d). \\
For each realization of the initial state, Hamiltonian, and coupling
operator, we compute $T^{-2}Q^{\mathrm{ctn}}_{\alpha\beta}=4T^{-2}\mathrm{Var}\left(G^{\mathrm{ctn}}_{\alpha\beta}\right){}_{\ket{\psi_{0}}}$
as follows: \\

\begin{enumerate}
\item As the initial state $\ket{\psi_{0}}$ is Haar-random,\textcolor{blue}{{}
}we draw its components independently from a standard normal distribution
and then normalize the resulting state, note this is the case of real
Haar-random states, for the complex case the components must instead
be drawn from a complex Gaussian distribution. 
\item Decomposition into the eigenbasis of $H_{\theta}$. It allows us to
perform the time integral analytically, so that the matrix elements
of the generator can be written as 
\begin{equation}
G^{\text{ctn}}_{\alpha\beta}=\Phi_{\alpha\beta}(T)P_{\alpha\beta}
\end{equation}
 where 
\begin{equation}
P_{\alpha\beta}\equiv\langle\lambda_{\alpha}|P|\lambda_{\beta}\rangle
\end{equation}
 are the matrix elements of $P={\displaystyle \sigma_{x}\otimes\mathbb{I}^{\otimes(L-1)}}$
in the eigenbasis of $H_{\theta}$ and 
\begin{equation}
\Phi_{\alpha\beta}(T)\equiv{\displaystyle \int}^{T}_{0}{\displaystyle ds\text{e}^{\text{i}(\lambda_{\alpha}-\lambda_{\beta})s}}
\end{equation}
\item Finally, we express $\ket{\psi_{0}}$ in the eigenbasis of $H_{\theta}$,
compute the variance of the generator, and obtain $T^{-2}Q^{\mathrm{ctn}}$
at time $T$ and average $T^{-2}Q^{\mathrm{ctn}}$ over different
realizations of initial states, Hamiltonians and possibly coupling
operators (in the randomized protocol) to obtain $T^{-2}\langle\!\langle\bar{Q}\rangle\!\rangle$
or $T^{-2}\langle\!\langle\tilde{\bar{Q}}\rangle\!\rangle$ . For
Figs.\ \ref{fig:Numerics}(b)\textcolor{blue}{{} }and\ \ref{fig:Numerics}(c),
we averaged over $500$ simultaneous realizations of initial states,
Hamiltonians and coupling operators (for the randomized protocol).
For Fig.\ \ref{fig:Numerics}(d), we averaged over $500$ simultaneous
realizations of initial states and coupling operators (for the randomized
protocol) without averaging over the Hamiltonians. The system size
is $N=2^{11}$ in all three panels (b-d) in Fig.\ \ref{fig:Numerics}.
In the deterministic case $P=\sigma_{x}\otimes\mathbb{I}^{\otimes(L-1)}$
, whereas in the randomized case $P$ is replaced by $V^{\dagger}PV$,
where $V$ is a matrix of Haar-random vectors stored column-wise. 
\end{enumerate}
\item Relative variance (RV) of the QFI show in Figs.\ \ref{fig:Numerics}(b)
and\ \ref{fig:Numerics}(c).\\
Using the numerical procedure above to compute different realizations
of $T^{-2}Q^{\mathrm{ctn}}_{\alpha\beta}$, we can compute the relative
variance at each time step. If the relative variance of a quantity
$O$ decreases with increasing system size, then the quantity $O$
is self-averaging. This property is particularly important for experiments
and numerical calculations, as it allows the number of required samples
to be reduced as the system size increases. In case of GOE and GUE,
the relative variance clearly decreases with increasing system size
(see inset in Fig.\ \ref{fig:Numerics}(b)), indicating that the
QFI is self-averaging at all timescales. For the disordered FRF-XXX
model, the relative variance of the randomized protocol exhibits behavior
similar to that observed in GUE, which decreases appreciably as the
system size increases, see inset in Fig\ \ref{fig:Numerics}(b).
On the other hand, for the deterministic protocol, the relative variance
shown in the inset of Fig.\ \ref{fig:Numerics}(c) is small.
\item Saturation values of the average QFI and its relative variance.\\
We also numerically calculate the saturation values of the average
QFI and the relative variance at long times ($T\sim10^{4}\sim10N$)
in Fig.\ \ref{fig:Numerics-1-2}. For GUE, GOE or the randomized
protocol in the FRF-XXX model, the saturation values of both the QFI
and the relative variance scale approximately as $N^{-1}$, which
agrees with very well the analytical prediction in the main text,
see Figs.\ \ref{fig:Numerics-1-2}(a,b,d) and Fig.\ \ref{fig:Numerics-1-2}(e,f,h).
In contrast, for the deterministic protocol, the saturation value
of the QFI scales approximately as $N^{-0.5}$, while that of the
relative variance scales approximately as $N^{0.15}$, see Figs.\ \ref{fig:Numerics-1-2}(c)
and Fig.\ \ref{fig:Numerics-1-2}(g). \\
It is worth to note that for the deterministic protocol, the slow
increase of the relative variance with the system size. As we have
mentioned in the main text, the deterministic protocol replies on
the assumption of statistical independence between the eigenvalues
and eigenvectors of the Hamiltonian. To what extend this assumption
holds is not clear in the literature. Here, if the approximation rigorously
holds, according to our proof of the late-time self-averaging property,
we expect the decrease of RV as the system size increases. Therefore,
it is plausible that the assumption of statistical independence may
not be a good approximation when computing RV. However, confirming
this conjecture requires computational resources beyond current capabilities. 
\end{itemize}
\begin{figure}
\centering
\includegraphics[scale=0.3]{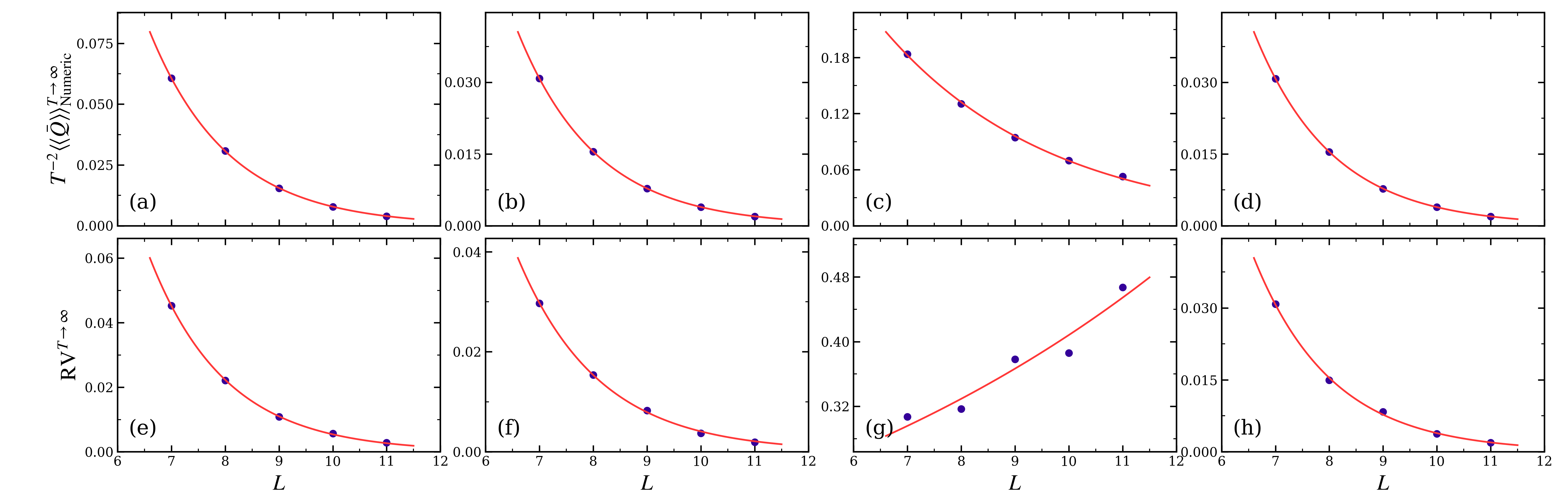}\caption{\label{fig:Numerics-1-2}\label{fig:Numerics-1}Saturation values
of the QFI and its relative variance for (a, e) GOE, (b, f) GUE, (c,
g) FRF-XXX model for the deterministic protocol, and (d, h) for the
randomized protocol. The blue dots correspond to the numerical results,
while the red curves represent fits of the form ${\displaystyle \propto N^{\nu}}$
with $N=2^{L}$. Value of parameters: $\theta=10^{-4}$; $500$ realizations
for $L=11$; $1000$ realizations for $L=10$; $1500$ realizations
for $L=9$; $3000$ realizations for $L=8$; and $5000$ realizations
for $L=7$.}
\end{figure}

\end{document}